\documentclass{PoS}

\usepackage{amssymb}
\usepackage{bm}
\usepackage{amsthm}
\usepackage{amsmath}
\usepackage{graphicx}
\usepackage{xcolor}
\usepackage{enumitem}   
\usepackage{comment}
\def\b#1{{\mathbb #1}}
\newcommand{\CC}{\mathbb{C}}
\newcommand{\RR}{\mathbb{R}}
\newcommand{\ZZ}{\mathbb{Z}}
\newcommand{\NN}{\mathbb{N}}

\newcommand{\Hi}{{\cal H}}
\newcommand{\Ob}{{\cal O}}
\newcommand{\Q}{{\cal Q}}

\newcommand{\A}{{\cal A}}
\newcommand{\B}{{\cal B}}
\newcommand{\R}{{\cal R}}
\newcommand{\Sy}{{\sf S}}
\newcommand{\T}{{\cal T}}
\newcommand{\Be}{\bm{e}}
\newcommand{\bu}{\bm{u}}
\newcommand{\bv}{\bm{v}}
\newcommand{\bx}{\bm{x}}
\newcommand{\Bp}{\bm{p}}
\newcommand{\bL}{\bm{L}}
\newcommand{\bpi}{{\bm \pi}}

\newcommand{\bomega}{{\bm\omega}}
\newcommand{\bphi}{{\bm\phi}}
\newcommand{\la}{\langle}
\newcommand{\ra}{\rangle}
\def\1{{\bf 1}}
\def\id{\mbox{id\,}}

\def\b{\mathfrak{b}}

\newcommand{\be}{\begin{equation}}
\newcommand{\ee}{\end{equation}}
\newcommand{\bea}{\begin{eqnarray}}
\newcommand{\eea}{\end{eqnarray}}
\newcommand{\ba}{\begin{array}}
\newcommand{\ea}{\end{array}}

\newtheorem{teorema}{Theorem}[section]

\newtheorem{corollario}{Corollary}[section]

\newtheorem{propo}{Proposition}[section]

\title{Energy cutoff, effective theories, noncommutativity, fuzzyness: 
 the case of $O(D)$-covariant fuzzy spheres}

\ShortTitle{Energy cutoff, noncommutativity and fuzzyness: 
~  the $O(D)$-covariant fuzzy spheres}

\author{\speaker{Gaetano Fiore},  Francesco Pisacane%         \thanks{A footnote may follow.}
\\
       Dip. di Matematica e Applicazioni, Universit\`a di Napoli ``Federico II'',\\
\& INFN, Sezione di Napoli, \\
Complesso Universitario  M. S. Angelo, Via Cintia, 80126 Napoli, Italy\\
        E-mail: \email{gaetano.fiore@na.infn.it},  \email{francesco.pisacane@unina.it}}

\abstract{Projecting a quantum theory onto the Hilbert subspace
of states with energies below a  cutoff
$\overline{E}$ may lead to an effective theory with modified
observables, including a noncommutative space(time). Adding a confining potential well $V$  
with a very sharp minimum on a submanifold $N$ of the original space(time) $M$
may  induce a dimensional reduction to a noncommutative quantum theory on $N$.
Here in particular we briefly report on our application \cite{FioPisJGP18,FioPis18POS,FioPis19JPA,FioPis19LMP,Pis20} of this procedure to 
spheres $S^d\subset\RR^D$ of radius $r=1$ ($D=d\!+\!1>1$):
making $\overline{E}$ and the depth of the well depend on (and diverge with) 
$\Lambda\in\NN$
we obtain new fuzzy spheres $S^d_{\Lambda}$  covariant under the {\it full} orthogonal 
groups $O(D)$; the commutators of the coordinates depend only on the angular momentum, as in Snyder noncommutative spaces. Focusing on $d=1,2$, we also discuss uncertainty relations,
localization of states, diagonalization of the space coordinates and construction of coherent states. As  $\Lambda\to\infty$ the Hilbert
space dimension diverges, $S^d_{\Lambda}\to S^d$, 
and we recover ordinary quantum mechanics on $S^d$.  
These models might be suggestive  for effective models in quantum field theory,  quantum gravity or condensed matter physics.}

\FullConference{Corfu Summer Institute 2019 "School and Workshops on Elementary Particle Physics and Gravity" (CORFU2019), \ 		31 August - 25 September 2019, \
		Corfu, Greece}

\begin{document}

%\tableofcontents

\section{Introduction}

The first example of noncommutative spacetime was proposed in 1947 by Snyder
\cite{Snyder} with the hope that  nontrivial (but Poincar\'e covariant)  commutation relations
among the coordinates  could 
cure  ultraviolet (UV) divergences in quantum field theory (QFT)\footnote{The idea had originated  in the '30s from Heisenberg, who proposed it in a letter to Peierls
\cite{Heisenberg30}; the idea propagated via Pauli to
Oppenheimer, who asked  his student  Snyder to develop it.}.
Shortly afterwards  the regularization of UV divergences based on an energy  cutoff, although  not Poincar\'e covariant, allowed
the renormalization of quantum electrodynamics; in the following decades this and other regularization methods within the renormalization program
have  allowed the extraction of physically
accurate predictions from quantum electrodynamics, chromodynamics, and 
the Standard Model of elementary particle physics. Therefore Snyder's model was almost forgotten for long time.
% (exceptions are e.g. \cite{Kad62,Mir67})
On the other hand, there is general consensus that any merging of quantum theory and general relativity
in an acceptable quantum gravity theory should lead to a cutoff (upper bound) on the local concentration of energy and to an associated lower bound 
(the Planck length  $l_p=\sqrt{\hbar  G/c^3}\sim 10^{-33}$cm) on the localizability of events. In fact,  by Heisenberg uncertainty relations,
to reduce the uncertainty $\Delta x$ of the coordinate $x$ of an event one must  increase
the uncertainty  $\Delta p_x$ of the conjugated momentum component by use of higher energy probes;
but by general relativity the associated concentration of energy in a small region  would produce 
a trapping surface (event horizon of a black hole) if it were too large; hence the size of this region, and $\Delta x$ itself, cannot be lower than the associated Schwarzschild radius, i.e.
 $l_p$. This heuristic argument \cite{Mea64} was made made more precise by 
 Doplicher, Fredenhagen, Roberts \cite{DopFreRob95}, who also  proposed that the latter bound 
could follow from appropriate noncommuting coordinates (for a  review of more recent developments see \cite{BahDopMorPia15}).

%Cita anche Zoupanos etc. As an arena for unifying interactions [Connes-Lott '92,...]
% connes etc

We  begin this paper observing  that in fact all these facts may stem from the same
(energy cutoff) mechanism: 
introducing an energy cutoff $\overline{E}$ in a quantum theory on a commutative space(time) $M$, i.e. projecting the theory on the Hilbert subspace with energy below $\overline{E}$,  directly induces a noncommutative deformation of the latter and lower bounds for the space(time) localizability.  Moreover, adding a confining potential well $V$  with a very sharp minimum on a submanifold $N$ of $M$
may  induce a dimensional reduction to a noncommutative quantum theory on $N$.
In \cite{FioPisJGP18,FioPis18POS,Pis20} we have applied this idea to obtain new fuzzy spheres $S^d_\Lambda$ of any dimension $d$ starting from quantum mechanics on
ordinary Euclidean spaces; while the seminal Madore-Hoppe fuzzy sphere (FS) \cite{Mad92,HopdeWNic} is covariant only under the rotation group, 
our $S^d_\Lambda$ are covariant  under the whole orthogonal groups.  After the mentioned general arguments, here we summarize how the $S^d_\Lambda$ are constructed and their main features, including uncertainty relations,
localization of states, diagonalization of the space coordinates and construction of coherent states \cite{FioPis19JPA,FioPis19LMP} for $d=1,2$. 

We recall that a fuzzy version of a commutative manifold $M$ is a sequence 
$\{\mathcal{A}_n\}_{n\in\NN}$
of {\it finite-dimensional} algebras such that
 $\mathcal{A}_n\overset{n\rightarrow\infty}\longrightarrow\mathcal{A}\equiv$algebra 
of regular functions on $M$. Since their introduction fuzzy spaces have raised a keen interest 
among mathematical and high-energy physicists as a non-perturbative technique in  QFT (or string, or M-, theory) based on a finite-discretization of space(time)  alternative to the lattice one; one main advantage  is that the algebras $\A_n$ can carry representations of Lie groups (not only of discrete ones).
In a  QFT  on a fuzzy space the ``cutoff'' $n$ works as 
a  parameter regularizing UV divergences, because  integration over fields  amounts to integration over matrices of a finite size, growing with $n$
(see e.g. \cite{GroMad92,GroKliPre96'} for the first QFT on the FS \cite{Mad92,HopdeWNic}, and \cite{GroKliPre96,Ramgoolam,Dolan:2003th,Ste17} for examples of QFT on fuzzy spheres of higher dimensions). If spacetime $M$ is enlarged to a higher-dimensional one $M'=M\times S_n$ - where $S_n$ is a fuzzy space, 
instead of a compact manifold $S$ - it reduces the number of massive Kaluza-Klein modes of a field theory on $M'$  to a finite value \cite{AscMadManSteZou,GavManOrfZou15} (the extra
dimensions can be used to describe internal degrees of freedom).
In the matrix model formulations of  $M$-theory \cite{Banks, Berko} and string theory
\cite{IKKT97} fuzzy spaces may arise as subalgebras 
 giving the leading contribution to the path-integrals over  larger matrix algebras;
they respectively lead to quantized branes  in a 11- or 10- dimensional spacetime.

\medskip
Consider a quantum theory $\T$; we denote the Hilbert space of the system $\Sy$
by \ $\Hi$, \ the algebra of observables on $\Hi$ (or with a domain dense in $\Hi$)
%drop?
  by  \ $\A\equiv\mbox{Lin}(\Hi)$, \ the Hamiltonian by \ $H\in\A$. \ 
For a generic subspace $\overline{\Hi}\subset\Hi$ let $\overline{P}:\Hi\mapsto\overline{\Hi}$ \ be  the associated projection and
$$
\overline{\A}\equiv \mbox{Lin}\left(\overline{\Hi}\right)
=\{\overline{A}\equiv \overline{P}A\overline{P}\:\: |\:
A\in\A\}\neq \A.
$$
Assume now  $\overline{\Hi}$ is a subspace  such that:
i) \ $\overline{P}H=H\overline{P}$; \   ii) \ $\overline{\A}$ \ contains all the observables corresponding to measurements that we can {\it really}
perform with the experimental apparati at our disposal.
If the initial state of the system belongs\footnote{If the state is not pure, but described by a density matrix $\rho$, the condition becomes "if $\rho\in\overline{\A}$".}  to 
$\overline{\Hi}$, then neither the dynamical evolution ruled by $H$, nor any measurement can map it out of $\overline{\Hi}$, and we can describe  $\Sy$ by
the effective theory \ $\overline{\T}$ \ based on  the projected Hilbert space $\overline{\Hi}$, algebra of observables $\overline{\A}$ and Hamiltonian
\ $\overline{H}= H|_{\overline{\Hi}}$. \ If $\overline{\Hi}$, $H$ are invariant under some  group $G$, then 
\ $\overline{P},\overline{\A},\overline{H},\overline{\T}$ \ will be as well.

\medskip
As a  particular consequence, {\bf if the theory $\T$ is based on commuting 
coordinates $x_i$
(commutative space) this will be in general no longer  true for $\overline{\T}$: 
\ $[\overline{x_i},\overline{x_j}]\neq 0$.}

\medskip
A physically relevant instance of the above projection mechanism occurs when
$\overline{\Hi}$ is the subspace of $\Hi$ characterized by energies  $E$ below
a certain cutoff, \ $E\le\overline{E}$; 
then $\overline{\T}$ is a {\it low-energy effective approximation} of $\T$. 
The prototypical example  
is Peierls projection \cite{Peierls_magn} (see also \cite{Jackiw,Magro}) applied to the Landau model of a charged particle in a plane subject to a perpendicular magnetic field $B$: \  choosing $\overline{E}$ equal to the ground state energy $E_0$   implies
\  $[\overline{x_1},\overline{x_2}]=\frac{i\hbar c}{ieB}$ (here $e$ is the electric charge
of the particle, $c$ is the speed of light, $x_1,x_2$ are the Cartesian coordinates of the particle on the plane), so that the effective theory is on a noncommutative space. 
$\overline{E}$ is a deformation parameter, in the sense $\overline{\T}\to\T$ as
$\overline{E}\to\infty$. If $H$ is $G$-invariant then 
also $\overline{\Hi}$ and therefore
\ $\overline{P},\overline{\A},\overline{H},\overline{\T}$ \ automatically  are.
Given an observable $A$ (e.g.  $A=x_1$, in the Landau model), $\overline{A}$ will
measure the {\it same physical quantitity} (the $x_1$ coordinate of the particle, in the mentioned example) {\it with
an uncertainty compatible with $E\le\overline{E}$}; in other words, the measurement process
cannot make the system jump  out of $\overline{\Hi}$, i.e. in states of 
 energy $E>\overline{E}$.

Imposing an energy cutoff $E\le\overline{E}$ on theory $\T$ may  be useful at least for the following reasons
(which may co-exist):
 
\begin{itemize}

\item If $\overline{\Hi}^\perp$ is practically not accessible in preparing the initial state, nor through the dynamical evolution (which may include interactions with the environment,  encoded in the possibly time-dependent Hamiltonian), nor through the measurement processes, then $\overline{\T}$ on the smaller Hibert space $\overline{\Hi}$  is 
in principle sufficient for determining all physical predictions and in fact simpler to work with.

\item If at $E>\overline{E}$ we expect new physics  
not accountable by $\T$, then  $\overline{\T}$ 
may also help to figure out a new theory $\T'$ valid for all $E$.

\item As a regularization procedure of a QFT $\T$, an energy cutoff $\overline{E}$ may allow to make sense of $\T$ if this is originally   ill-defined  due to UV divergences - e.g. divergent contributions (loop 
integrals) to transition amplitudes - for generic finite values of (a finite number of)  bare parameters $\mu_I$ (e.g. masses, coupling constants,...)  present in the Hamiltonian/Action. These divergent contributions are due to virtual intermediate states of arbitrarily high energy $E$ that can be assumed by the system during the interaction. Imposing  $E\le\overline{E}$ (or some other regularization scheme) allows to  make  the (unknown) $\mu_I$
well-defined (at least in a perturbative sense) functions $\mu_I(\Q,\overline{E})$ of a small number of observable quantities $\Q_i$ (e.g. masses of asymptotic states, large distance coupling constants,...) and of 
$\overline{E}$ (or the other regularization parameter). Replacing these functions in the dependences $\Ob_A(\mu)$ of all the observables $\Ob_A$ (e.g. cross sections in scattering processes, decay times of unstable particles, etc.; here $A$ stands for a collective index which allows to distinguish not only the type of observable, but also the involved initial and final data, e.g. the initial and final type, number, momenta of the particles involved in a scattering process)  on the $\mu_I$ yields functions $\overline{\Ob}_A(\Q,\overline{E})$.
If the latter admit $\overline{E}\to \infty$ limits the theory is said to be UV renormalizable, and 
these limits tipically give a physically accurate relation between $\Q$ and the observed 
$\overline{\Ob}_A$.

 \end{itemize}

\begin{figure}[t]
\includegraphics[width=16cm]{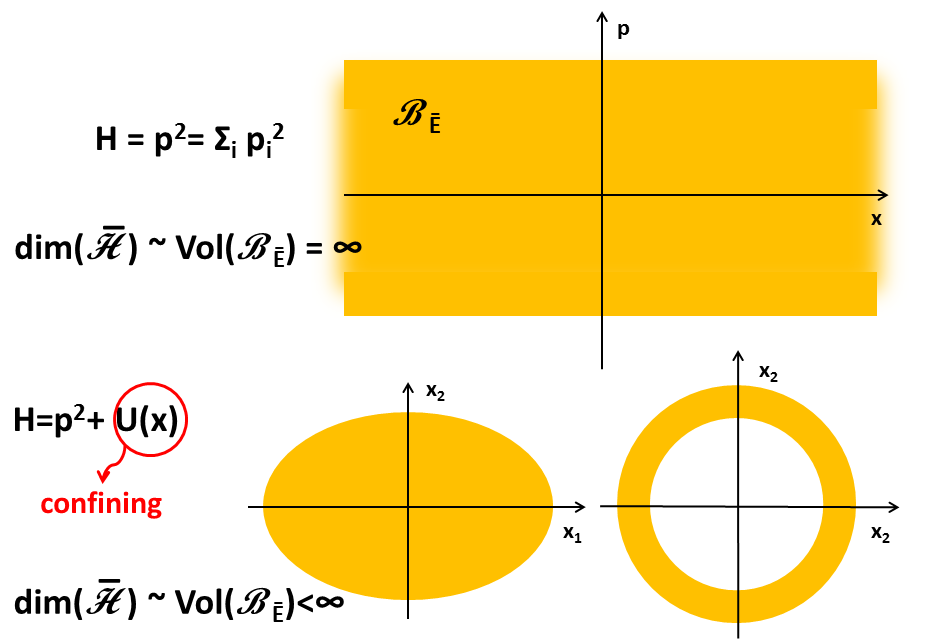}
\caption{Up: Classically allowed  phase space region $H(\bx^c,\Bp^c)\le\overline{E}$
if $H=\Bp^2$. \ Down: 
Classically allowed  configuration space region if $H=\Bp^c{}^2+V\big(\bx^c\big)$, with
the potential of the form $V\big(\bx^c\big)=a(x_1^c)^2+b(x_2^c)^2$ (left) or $V\big(\bx^c\big)=2k(r_c-1)^2$ (right), where $r_c=\sqrt{\bx^c{}^2}$.}
\label{fig1} 
\end{figure}

\bigskip
As a $\T$ consider now quantum mechanics (QM) of a zero spin particle on $\RR^D$ with a Hamiltonian $H(\bx,\Bp)$. 
If $\overline{\Hi}$ is  the subspace with energies $E\le\overline{E}$ then
its dimension is approximately  the  phase space volume of
the classical region  $\B_{\scriptscriptstyle\overline{E}}$ determined by the inequality 
$H\big(\bx^c,\Bp^c\big)\le\overline{E}$, in Planck constant units:
$$
\mbox{dim}(\overline{\Hi})\simeq\mbox{Vol}(\B_{\scriptscriptstyle\overline{E}})/h^D.
%\qquad \B_{\overline{E}}\equiv\{(x,p)\in\RR^{2D}\:\:|\:\: H(x,p)\le \overline{E}\}
$$
This is still infinite if  e.g. $H$  reduces to the kinetic term $\Bp^2$ (upper part of fig.
\ref{fig1}), while it is finite if $H$ contains a sufficiently strong 
binding potential $V(\bx)$
(lower part of fig. \ref{fig1}); consequently $\overline{\T}$  will be  a {\bf  
fuzzy approximation} of QM approximately confined in the (configuration space) region $\R\subset\RR^D$
determined by the inequality $V(\bx)\le \overline{E}$.\footnote{Of course, one can obtain a 
fuzzy noncommutative approximation of QM in a region $\R$ also imposing an energy
cutoff on a pre-existing noncommutative deformation of QM on $\R$, see e.g. the fuzzy disc of \cite{LizVitZam03}.}
We can obtain a  {\bf  NC, 
fuzzy approximation}  of QM on a {\bf submanifold} $N$ of $\RR^D$
adding a  {\bf `dimensional reduction' mechanism}, more precisely a $V(\bx)$
with a sharp minimum on $N$.\footnote{In passing, we note that  
defining submanifolds of noncommutative spaces is delicate problem \cite{Dan19};\cite{FioWeb20} proposes a rather general procedure in the framework of
Drinfel'd twist deformations of differential geometry.}

In the rest of the paper we report on our application \cite{FioPisJGP18,FioPis18POS,FioPis19JPA,FioPis19LMP,Pis20}
of the mechanism for $N$ equal to the  $d=(D\!-\!1)$-dimensional
 {\it sphere $S^d$} of radius $r=1$ ($r^2:=\bx^2$ is the square distance from the origin) and on the study of the resulting fuzzy spheres for $d=1,2$ \cite{FioPisJGP18,FioPis18POS,FioPis19JPA,FioPis19LMP}; 
the lower right corner of  fig. \ref{fig1} shows the corresponding 
region $\R$ (a thin spherical shell of radius $\simeq 1$) in the $d=1$   case. 
The plan  is as follows. Section \ref{Preli} contains
further preliminaries. In section \ref{GenConst} we sketch the construction
procedure of $S^d_\Lambda$ for generic $d\ge 1$.  In sections \ref{D2}, \ref{D3} we  respectively review   the 
 main features of $S^1_\Lambda,S^2_\Lambda$, 
 the eigenvalues and eigenvectors of the associated coordinate operators
$x_i$ (for the latter we prove slightly stronger results than in \cite{FioPis19JPA}); then we  present various systems of coherent states (SCS) on them and discuss their
localization both in configuration and (angular) momentum space. 
Finally, in section \ref{Conclu} we draw the conclusions and add final remarks, while 
comparing our $S^d_\Lambda$ with other fuzzy spheres, in particular
the celebrated Madore-Hoppe Fuzzy sphere (FS) \cite{Mad92,HopdeWNic}.

\section{Further preliminaries}
\label{Preli}

\subsection{Covariance}

$O(D)$-covariance of the theory means that for any orthogonal matrix $g\in O(D)$ there is
a unitary transformation $U(g)$  of the Hilbert space $\Hi$ (or $\overline{\Hi}$)
%commuting with scalar observables and 
such that $g_{ij}v_j=U^{\dagger}\!(g)\, v_i\, U(g)$ for all vectors $\bv$, and similarly for
other $O(D)$-tensors.
%, in particular $\bv=\bx$.
Fixed a $\bv$,  we  can split
\bea
\Hi=\bigcup_{\bu\in S^d}\Hi_{\bu}, \qquad \qquad\Hi_{\bu}:=\left\{\psi\in\Hi \:\:\: | \:\:\:
%\la \bv\ra_\psi\propto\bu
\left\la \bv\right\ra_\psi=\left|\la \bv\ra_\psi\right|\,\bu\,
\right\} .                        \label{splitH}
\eea
For each (unit vector) $\bu\in S^d$ consider a  $g\!\in\! O(D)$ such that $g\bu\!=\!\Be_1$, where $\Be_1\!:=\!(1,\!0,\!...,\!0)$, and define \ $\bv'\!:=\!g\bv$, \
so that $\bv\cdot\bu=v_1'$. For all $\psi\in\Hi_{\bu}$ \ we find \ 
$\left\langle \bv'\right\rangle_\psi=|\la \bv\ra_\psi|\Be_1$; \ moreover, \
$ \Hi_{\bu}=U^{\dagger}(g)\Hi_{\Be_1}$
(of course one obtains the same result replacing $\Be_1$ by any other $\Be_i$).

\subsection{Localization on $\mathbb{R}^D$, $S^d$ and $S^d_\Lambda$}
\label{Localization}

\begin{figure}[t]
%\begin{center}
\includegraphics[width=9cm]{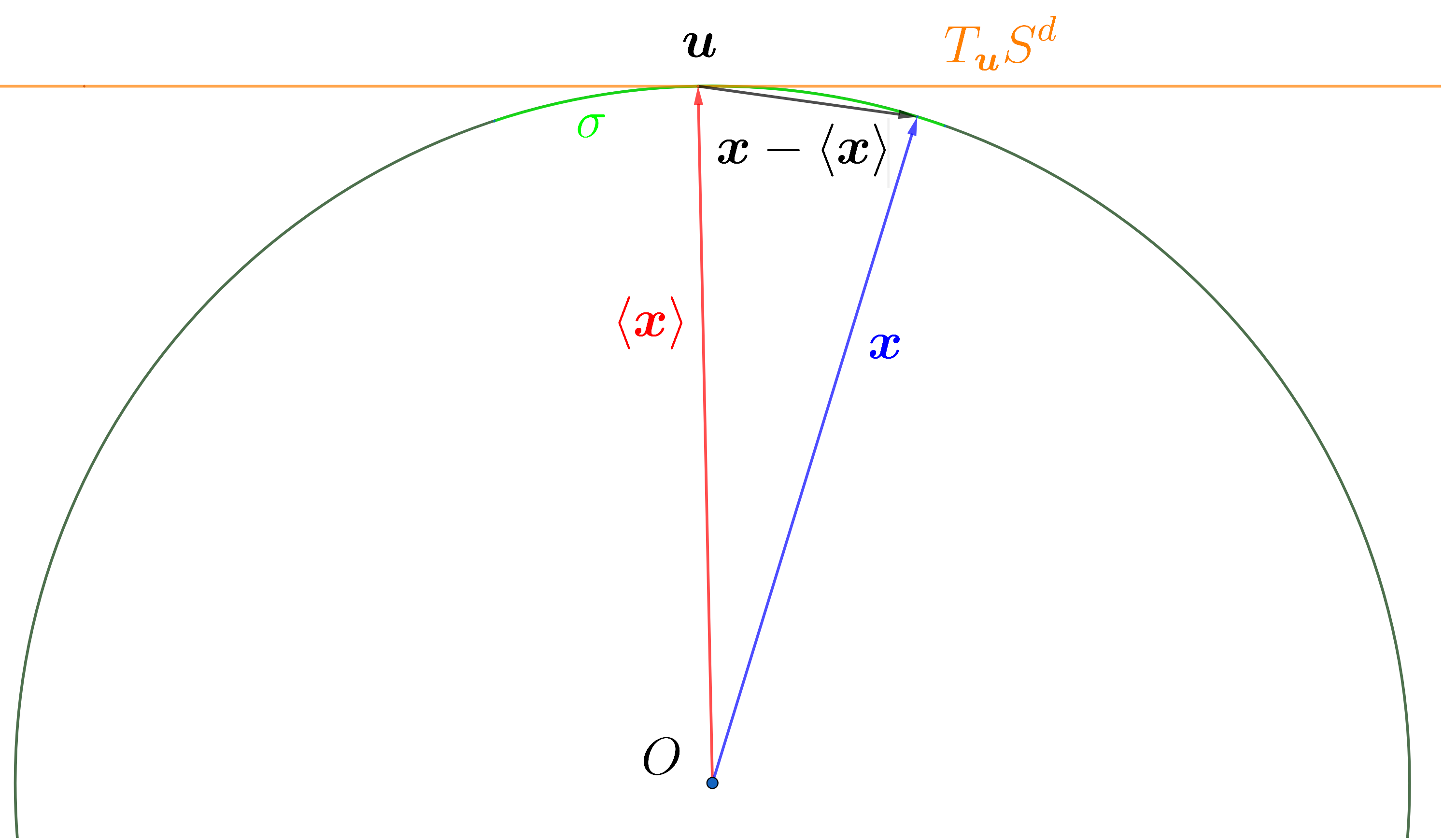}
\hfill\includegraphics[width=5cm]{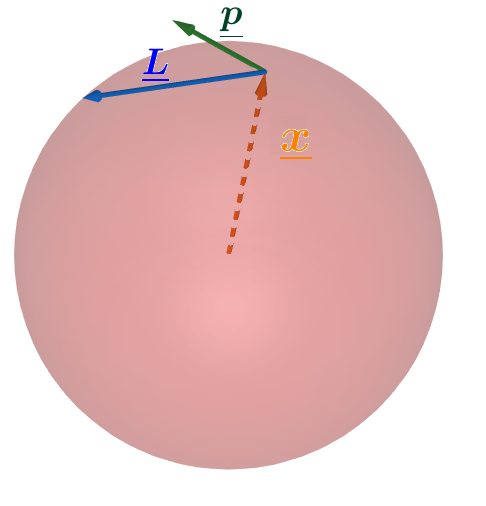}
%\end{center}
\caption{Left: the vectors {\color{blue}$\bm{x}$}, {\color{red}$\bu\equiv\left\langle\bm{x}\right\rangle$}, $\bm{x}-\left\langle\bm{x}\right\rangle$, the region {\color{green}$\sigma$} and the tangent plane {\color{orange}$T_{\bm{u}}S^d$} at $\bm{u}$. \ Right: perpendicularity
of $\bx$ and $\bL$.}
\label{Vett_tg}
\end{figure}

A good measure of the localization of a state in configuration space $\mathbb{R}^D$
is  its {\it spacial dispersion}, i.e.  the $O(D)$-invariant (and therefore reference-frame-independent) expectation value 
\be\label{uncnostra}
(\Delta \bx)^2\equiv\sum_{i=1}^D(\Delta x_i)^2
\equiv \left\langle \left( \bm{x}\!-\!\left\langle\bx\right\rangle\right) ^2\right\rangle
=\left\langle\bx^2\right\rangle- \left\langle \bx\right\rangle^2 
%=\left\langle \sum_{i=1}^Dx_ix_i\right\rangle-\sum_{i=1}^D\left\langle x_i\right\rangle^2;
\ee
on the state. Here  $\bm{x}\equiv(x_1,...,x_n)$ is the vector position observable of the particle
 in the ambient Euclidean space  $\RR^D$, the vector
$\left\langle \bx\right\rangle\equiv(\left\langle x_1\right\rangle,...,\left\langle x_n\right\rangle)$ pinpoints  the average position,
the scalar observable $\bx^2:=\sum_{i=1}^Dx_ix_i$ \
measures the square distance from the origin, 
the vector observable 
$\bx\!-\!\left\langle\bx\right\rangle$ 
measures the displacement from $\left\langle \bx\right\rangle$;  
(\ref{uncnostra}) is the expectation value of the square of the latter. 
We adopt $(\Delta \bx)^2$ also on $S^d,S^d_\Lambda$: in fact, if  the state  is localized in a small region 
$\sigma\subset S^d$ 
around a point $\bu\equiv\left\langle\bx\right\rangle\in S^d$ then
$(\Delta \bx)^2$ essentially reduces to the average square displacement
in the tangent  plane at $\bu$ (see fig. \ref{Vett_tg}, left), as wished. 
 If $\bx^2\equiv 1$ on the whole Hilbert space $\Hi$ 
(this occurs strictly if $\Hi={\cal L}^2(S^d)$
and also on Madore's FS  $S^2_n$, only approximately on our $S^d_\Lambda$), then $\left\langle \bx^2\right\rangle$
is state-independent, and  (\ref{uncnostra}) is minimal on the states
with maximal $\left\langle \bx\right\rangle^2$. By (\ref{splitH}) with $\bv=\bx$,
in each $\Hi_{\bu}$ $\left\langle \bx\right\rangle^2$ 
is maximized on the eigenvector(s) $\psi$ of  $x_1'=\bx\cdot\bu$
with the highest (in absolute value) eigenvalue (the latter exists on the Madore's FS, 
while  on $S^d$ it exists 
as a generalized eigenvector).

%This occurs in Madore's FS in a strict sense; on ${\cal L}^2(S^d)$...

\subsection{Diagonalization of a coordinate \ $x_i$ , \ and most localized states}
\label{diagx}

For $x_i$ to approximate well and $O(D)$-covariantly a coordinate of a quantum particle forced to stay on the commutative sphere $S^d$, its spectrum $\Sigma_{x_i}$ should fulfill at least the following properties:

\begin{enumerate}
\item $\Sigma_{x_i}$ is the same for all $i=1,...,D$ and choices of the reference frame. In particular, it is invariant under inversion \ $x_i\mapsto -x_i$.
\item In the commutative limit  $\Sigma_{x_i}$  becomes uniformly dense in $[-1,1]$, in particular the maximal and the minimal eigenvalues  converge to $1$ and $-1$, respectively.
\end{enumerate}
These properties are fulfilled by both the Madore FS and (at least for $d=1,2$) our  $S^d_\Lambda$. As explained in the previous subsection, the eigenstates with maximal eigenvalue
(in absolute value) have also maximal localization on $S^d ,S^2_n$; this also approximately 
ture on our $S^d_\Lambda$.

\subsection{Systems of coherent states (SCS)}
\label{SCS}

\noindent 
We recall that the canonical SCS  \ $\{\phi_z\}_{z\in\Omega}\subset\Hi$  on $\RR^D$   can be defined in three equivalent ways:
%(at least)

\begin{enumerate}

\item As the set of states saturating Heisenberg uncertainty relations (HUR) 
\ $\Delta x_i\Delta p_i\ge 1/2$. 

\item As the set of eigenstates of all annihilation operators 
$a_i$  with set of joint eigenvalues $z\in\Omega$.
%; on other manifolds $M$ nontrivial problem!

\item As the set of states generated by the group $G$ acting on the vacuum state $\phi_0$. 

\end{enumerate}
Here $\Hi={\cal L}^2(\RR^D)$, $\Omega\!\equiv\!\CC^D$,
 all variables have been made  dimensionless,  $a_i=x_i+i p_i$,
and $G$ is the Heisenberg-Weyl group. Characterizations 1 (for $D=3$), 2, 3 are
 due to
Schr\"odinger himself and   Klauder, Sudarshan, Glauber \cite{Schroe26,Klauder,Sudarshan,Glauber}. 
All of them admit (in general, non-equivalent) generalizations; 
see e.g. \cite{Klauder-Skager85,AliAntGaz13,AntBagGaz18,FioGueMaiMaz15,FioJPA18,FioRM19},
also for an overview on applications in elementary particle, nuclear, atomic, condensed matter, plasma  physics. The canonical SCS fulfills the following properties:
\begin{enumerate}%[(a)]
[label=\alph*)]
\item    {\bf Strong continuity} of $\phi_z $ as a function of $z\in\Omega$;

%\textbf{Resolution of the identity}:  \
%$\exists\, dz\equiv$ integration measure  on $\Omega$:
\item 
\textbf{Resolution of the identity: } $\quad
\id=\int_{\Omega} \, d\!\mu(z)\: P_z, \qquad P_z =\phi_z  \langle \phi_z,\cdot\rangle 
\equiv|\phi_z\rangle\langle\phi_z|; %\:\:\label{ResolId0}
$

\item     \textbf{Completeness}: \ \
$\overline{\mbox{Span}\left\{\phi_z \: |\: z\in\Omega\right\}}=\mathcal{H}$. 
\end{enumerate}
where $d\mu(z)=d\Re(z)\,d\Im(z)$, and the resolution b) is in the weak sense. 
These properties are often used 
\cite{Klauder-Skager85} for  defining SCS in general: 
a set $\{\phi_z\}_{z\in\Omega}\subset\Hi$, where $\Omega$ is a topological label space,
 is a {\it strong} SCS if it fulfills 
a), b) with a suitable  integration measure $d\!\mu(z)$ on $\Omega$;
a {\it weak} SCS if it fulfills a), c). \ As b) implies c), a strong SCS is also weak.
Perelomov  and  Gilmore develop \cite{PerelomovCMP72,GilmoreAnPh72} the concept of SCS  through approach 3 choosing $\Omega$ either a generic 
%locally compact ?
Lie group $G$, or more generally a coset $G/H$ thereof, acting  on $\Hi$ via an irreducible unitary representation $T$ (see e.g. \cite{Perelomov}). The steps are as follows:

\begin{itemize}
\item For all \  $\phi_0\!\in\!\Hi$, \
 let \ $\phi_g\!\equiv\! T(g)\phi_0$  for all $g\!\in\! G$, \ %maximal subgroup 
$H\!\equiv\!\{h\!\in\! G\:\: |\:\: \phi_h =\exp{\left[i\alpha(h)\right]}\phi_0\}$.

\item Then \ $|\phi_g\rangle\langle\phi_g|=|\phi_{gh}\rangle\langle\phi_{gh}|\equiv P_z$, 
\ i.e. depends only on \ $z\equiv[g]\in G/H\equiv\Omega$.

\item If  $\phi_0$ is {\it admissible}, i.e.    $%I_T\equiv
\int_{G}|\langle\phi_0,\! T(g)\phi_0\rangle|^2\,dg<\infty$, where  $dg$
is the left-invariant Haar measure on $G$, then  b) holds with  $d\!\mu(z)$
the normalized measure induced by $dg$ on $\Omega$. 
\end{itemize}
Clearly, if $G$ is compact all $\phi_0\!\in\!\Hi$ are admissible.
Following
Perelomov, the  CS that are closest  to classical states are obtained from a $\phi_0$ that maximizes $H$, or better the  isotropy subalgebra $\b$
in the {\it complex} hull of the Lie algebra of $G$;  $\phi_0$  is annihilated by some element(s) in $\b$,  the corresponding $\phi_g$ are eigenvectors of the latter
(property 2) and minimize the $G$-invariant uncertainty associated to the quadratic Casimir
[$(\Delta \bm{L})^2=\sum_{i<j} \Delta L_{ij}^2$  in the case $G=SO(D)$].
For $G=SO(3)$ it is $H=SO(2)$, $(\Delta \bL)^2=\left\langle\bL^2\right\rangle- \left\langle \bL\right\rangle^2$ (with \ $L_i\equiv\varepsilon^{ijk}L_{jk}/2$),  \ and 
%we have shown that 
minimizing $(\Delta \bm{L})^2$ amounts to saturating a specific UR \cite{FioPis19LMP}
(hence also
property 1 holds); this SCS consists of the socalled {\it coherent spin} %\cite{Rad71} 
or {\it Bloch} states.

In introducing SCS  on $S^d_\Lambda$ ($d=1,2$) we follow in spirit Perelomov's approach, with $G$ the isometry group $O(D)$ of $S^d$ (a compact group). However, our Hilbert space 
$\mathcal{H}_{\Lambda}$
will in general carry a {\it reducible} representation  of $O(D)$; moreover,
we study the localization properties of these SCS  both in configuration and  (angular) momentum space.

\section{Construction of the  $S^d_\Lambda$ for general $d\ge 1$}
\label{GenConst}
%\vspace{-0.5cm}
%\begin{figure}[htbp]
%\includegraphics[width=8cm]{IMG_1830.JPG}
%\caption{Three-dimensional plot of $V(r)$}\label{fig1}
%\end{figure}
  \begin{figure}[t]
        \begin{minipage}[c]{.48\textwidth}
          \includegraphics[scale=0.19]{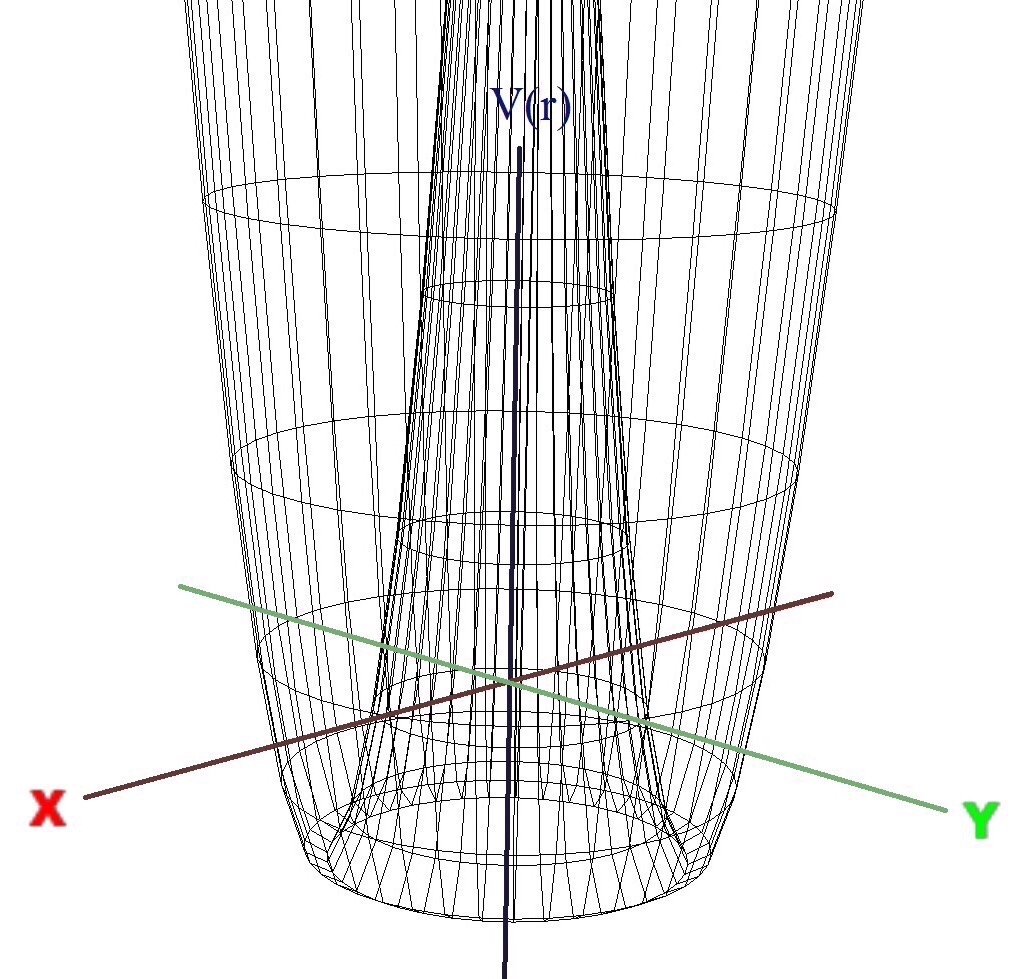}
          \caption{Three-dimensional plot of $V(r)$}
        \end{minipage}%
        \hspace{5mm}%
        \begin{minipage}[c]{.48\textwidth}
\begin{center}          \includegraphics[scale=0.35]{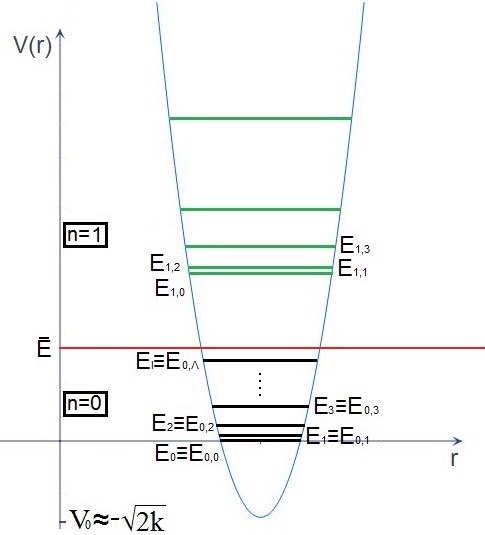}
          \caption{Two-dimensional plot of $V(r)$ including the energy-cutoff
and allowed energy levels (black).}
\end{center}        \end{minipage}
\label{fig2}
      \end{figure}

The main steps of the costructions are as follows:

\begin{itemize}

\item We adopt a $O(D)$ invariant Hamiltonian 
\begin{equation}
H=-\frac{1}{2}\Delta+V(r),
\end{equation}
where the confining potential $V(r)$ has a very sharp minimum $V_0=V(1)$ at $r=1$.
More precisely, we assume that 
\be
V(r)\simeq V_0+2k\left(r\!-\!1\right)^2\qquad\mbox{if } V(r)\leq \overline{E},   \label{cond}
\ee
so that $V(r)$ has a harmonic behavior for \ $\left\vert r\!-\!1\right\vert\!\leq\!\sqrt{\!\frac{\overline{E}\!-\!V_0}{2k}}$, and that $V''(1)\!\equiv\! 4k \!\gg\! 0$ 
($k$ thus parametrizes the sharpness of the minimum); we fix $V_0$ so that the ground state has energy $E_0=0$. Using polar coordinates we can decompose $\Delta=\partial_r^2+
\frac dr\partial_r-\frac 1{r^2}\bL^2$, where  $\bL^2:=L_{ij}L_{ij}/2$
is the square angular momentum [$L_{ij}:=i(x_j \frac{\partial}{\partial x_i}\!-\!x_i \frac{\partial}{\partial x_j})$ are the angular momentum components], i.e.
%$\left\{j\left(j+d-1\right)\right\}_{j=0}^{\Lambda}$ 
the Hamiltonian of free motions (the Laplacian) on $S^d$.
Looking for the eigenfunctions $\psi$ of $H$ in the form $\psi=f(r)Y(\varphi,...)$,
where $\varphi,...$ are the angular coordinates,
we reduce the eigenvalue equation $H\psi=E\psi$ to a 1-dimensional Schr\"odinger equation
in the form of an ordinary differential equation
 with respect to $r$. The eigenvalues are parametrized by integers $l,n\ge 0$; 
they respectively determine the eigenvalue \ $E_l\equiv l(l\!+\!d\!-\!1)$ \ of $\bL^2$ and the radial excitation,
which at least for small $n$ are approximately of harmonic type, \
$\simeq \sqrt{8}kn$. 

\item We choose $\overline{E}$ low enough, e.g. $\overline{E}\lesssim\sqrt{8}k$,  to constrain $n$ to be zero, namely to {\it eliminate radial excitations}
from the spectrum $\Sigma_{\overline{H}}$ of $\overline{H}$, 
so that the latter reduces to that of  $\overline{\bL^2}$, \  $\Sigma_{\overline{H}}=\{E_l\}$.
Then we also find that the {\bf $x_i$ generate the whole algebra of observables 
$\overline{\A}$}, and \   $[\overline{x_i},\overline{x_j}]\simeq - iL_{ij}/k$,  \
i.e. we find   Snyder-type commutation relations among the coordinates\footnote{Snyder's quantized spacetime algebra is generated by 4 hermitean Cartesian coordinate operators $\left\{x^\mu\right\}_{\mu=0,1,2,3}$, 
 and 4 hermitean momentum operators $\left\{p_\mu\right\}_{\mu=0,1,2,3}$  fulfilling (here $\alpha$ is a suitable constant)
\begin{equation}
[p_\mu,p_\nu]=0, \qquad
[x^\mu,p_\nu]=i\hbar(\delta^\mu_\nu-\alpha p^\mu p_\nu), \qquad [x^\mu,x^\nu]=-i\hbar\alpha L^{\mu\nu},\qquad\qquad
\mu,\nu=0,1,2,3
\end{equation}
where  $L^{\mu\nu}
=x^\mu p^\nu-x^\nu p^\mu$ and
$v^\mu=\eta^{\mu\nu}v_\nu$, with $\eta=\mbox{diag}(1,-1,-1,-1)=\eta^{-1}$
the Minkowski metric matrix.
%These relations are invariant under inversion of the axes, 
%in particular under parity. The $L^{\mu\nu}$ span the Lorentz Lie algebra,
%and their commutation relations with the 4-vectors $p^\mu,x^\mu$
%are as on the Minkowski space.
}. 
There is a residual freedom in the choice of 
$V(r)$ (the higher order terms in the Taylor expansion 
of $V(r)$ around $r=1$); we fine-tune the model requiring that  
$[\overline{x_i},\overline{x_j}]= - iL_{ij}/k$ (up to terms that act non-trivially
only on the highest energy states).

\item To obtain a sequence of finite-dimensional models going to QM on $S^d$
we make $\overline{E}$ grow and diverge with a natural number $\Lambda$; so must
also $k$ do, in order that the above inequality keeps holding. We choose \ $\overline{E}\equiv E_{\Lambda}=\Lambda(\Lambda\!+\!d\!-\!1)$ 
and $V$ depending on $\Lambda$ so that $k(\Lambda) \ge\Lambda^2(\Lambda\!+\!1)^2$; 
\ correspondingly, $\Sigma_{\overline{H}}=\{E_l\}_{l=1}^\Lambda$, and replacing everywhere the bar by the subscript $\Lambda$ we find
\bea
(\Hi_\Lambda,\A_\Lambda) \stackrel{\Lambda\to\infty}{\longrightarrow}
\left(\Hi,\A\right)\equiv \left({\cal L}^2(S^d),\,\mbox{Lin}\!\left({\cal L}^2(S^d)\right)\right) \label{LambdaLimit}
\eea
in a suitable sense \cite{FioPisJGP18}.
$\{S^d_\Lambda\}_{\Lambda\in\NN}\equiv \{(\Hi_\Lambda,\A_\Lambda)\}_{\Lambda\in\NN}$ is our  {\bf $d$-dimensional, $O(D)$-covariant fuzzy sphere},  
i.e. a sequence of finite-dimensional approximations of ordinary QM on 
$S^d$.

 \end{itemize}

\smallskip
It turns out that (at least for $D=2,3$) there exist $O(D)$-covariant $*$-algebra isomorphisms 
\ $\A_{\Lambda}%\simeq M_N(\CC)
\simeq\pi_\Lambda[Uso(D+1)]$, \
where $(\pi_{\Lambda},\Hi_\Lambda)$ is a suitable irreducible  unitary representation  of $Uso(D+1)$. More precisely, in terms of the canonical basis \
$\{{\sf L}_{IJ}\,| \, 1\le I<J\le D\!+\!1\}$ \ of $so(D+1)$, 
\be
\overline{L_{ij}}=\pi_{\Lambda}({\sf L}_{ij}),\qquad
\overline{x_h}=\pi_{\Lambda}\left[f_1\!\left(\bL^2\right)\,{\sf L}_{h({D+1})}
\,f_2\!\left(\bL^2\right)\right],\qquad 1\le i,j,h\le D, \quad i<j,    \label{isomD}
\ee
where and $f_1(s),f_2(s)$ are suitable analytic functions. 
%(check o di $\lambda=\sqrt{\bL^2}$?)

\smallskip
To simplify the notation, below we shall remove the bar and denote the generic  $\overline{A}\in\mathcal{A}_{\Lambda}$   as $A$.

\section{$D=2$: $O\left(2\right)$-covariant fuzzy circle}
\label{D2}

In a suitable orthonormal basis $\B_{\Lambda}:=\{\psi_\Lambda,\psi_{\Lambda-1},...,\psi_{-\Lambda}\}$ of the Hilbert space $\Hi_\Lambda$ consisting of eigenvectors of the angular momentum $L\equiv L_{12}$,
\be
L\, \psi_n=n\,\psi_n,
\ee
the action of the noncommutative coordinates \ $x_\pm:=x_1 \pm ix_2$
 \ of the fuzzy circle $S^1_\Lambda$ read\footnote{Here we use the conventions 
of \cite{FioPis19JPA,FioPis19LMP}, rather than those of   \cite{FioPisJGP18}.}
\bea
%L \psi_n=n\psi_n, \qquad  
x_{\pm}\,\psi_n \: =\:
\left\{ \begin{array}{ll}\!\! \displaystyle
 \left[1\!+\!\frac{n(n\pm1)}{2k}\right]\psi_{n\pm1} \qquad &
\mbox{if }-\!{{\Lambda}}\leq \pm n\leq{{\Lambda}}\!-\!1, \\[8pt]
0 & \mbox{otherwise,}
\end{array}\right.            \label{defLxiD2}
\eea
where $k=k(\Lambda)\ge\Lambda^2(\Lambda\!+\!1)^2$. In the $\Lambda=\infty$ limit
$x_\pm=e^{\pm i\varphi}$, $\psi_n=e^{in\varphi}$ (up to a phase); $\varphi$ is the angle along $S^1$. 
$L ,x_+,x_-$ and $\bx^2:=x_1 ^2+x_2^2=\frac 12(x_+x_-+x_-x_+)$ fulfill the $O(2)$-equivariant  relations
\be
\left[L , x_{\pm}\right]=\pm x_\pm,\quad x_+{}^\dagger=x_-, \qquad L ^\dagger=L,\label{commrelD=2'}
\ee
\be
\left[x_+,x_-\right]={\color{red}{ -\frac{2L }k+\left[1\!+\!\frac {\Lambda(\Lambda\!+\!1)}k\right]\!\left(\widetilde P_{\Lambda}\!-\!\widetilde P_{-\Lambda}\right)\equiv L'}},\label{y+y-}
%\left[x_+,x_-\right]=-\frac{2L }k+\left[1\!+\!\frac {\Lambda(\Lambda\!+\!1)}k\right]\!\left(\widetilde P_{\Lambda}\!-\!\widetilde P_{-\Lambda}\right),\label{y+y-}
\ee
\be
\bx^2=  1 \:{\color{red}{ +\frac{L ^2}{k} -
\left[1\!+\!\frac {\Lambda(\Lambda\!+\!1)}{k}\right]
\frac{\widetilde P_{\Lambda}\!+\!\widetilde P_{-\Lambda}}2}},          \label{defR2D=2}
%\bx^2=  1+\frac{L ^2}{k} -\left[1\!+\!\frac {\Lambda(\Lambda\!+\!1)}{k}\right]
%\frac{\widetilde P_{\Lambda}\!+\!\widetilde P_{-\Lambda}}2,          \label{defR2D=2}
\ee
\be
\prod\limits_{n=-\Lambda}^{\Lambda}\!\!\left(L \!-\!n\,I\right)=0, \qquad \left(x_\pm\right)^{2\Lambda+1}=0 .\label{commrelD=2}
\ee
Here $\widetilde{P}_n$ is the projection onto the $1$-dim subspace  
$\CC\psi_n$. Terms marked red are absent in the commutative case. In the $\Lambda\to\infty$
limit also the non-vanishing ones will play no role at any fixed energy $E$, as they are proportional to the projections $\widetilde P_{\pm\Lambda}$
onto the states with highest energy $E_{\Lambda}\to\infty$; 
(\ref{commrelD=2}a) gives back  $\Sigma_{L}=\ZZ$, whereas (\ref{commrelD=2}b) looses meaning and must be dropped.
We point out that: 

\begin{itemize}

\item $\bx^2\neq 1$, but it is a function of $L^2$, hence the $\psi_n$ are
its eigenvectors; its eigenvalues  
(except on $\psi_{\pm\Lambda}$)
are close to 1, slightly grow with $|n|$ and collapse to 1 as $\Lambda\to \infty$.

\item The ordered monomials  $x_+^h\,L ^l\, x_-^n$ [with degrees $h,l,n$
bounded by (\ref{commrelD=2'})-\ref{commrelD=2}] make up a 
 basis of the $(2\Lambda\!+\!1)^2$-dim  vector space 
underlying the algebra of observables  \ $\A_\Lambda\!:=\!End(\Hi_\Lambda\!)$
\ (the $\widetilde{P}_n$ themselves can be expressed as polynomials in $L $). 

\item $x_+ ,x_-$ generate the whole $*$-algebra $\A_\Lambda$,
because also $L $ can be expressed as a non-ordered polynomial in $x_+,x_-$. 

\item As anticipated in (\ref{isomD}), actually  there are $O(2)$-equivariant
$*$-algebra isomorphisms $\A_{\Lambda}$ 
 \be
\A_{\Lambda}\simeq M_N(\CC)\simeq\pi_\Lambda[Uso(3)], \qquad 
N=2\Lambda\!+\!1,                                 \label{isomD2}
\ee
where $\pi_{\Lambda}$ is the $N$-dimensional unitary irreducible representation  of $Uso(3)$.
The latter is characterized by the condition $\pi_{\Lambda}(C)=\Lambda(\Lambda+1)$, where  $C=E_aE_{-a}$ is the Casimir (sum over $a\in\{+,0,-\}$), and $E_a$
make up the Cartan-Weyl basis of $so(3)$,
\be
[E_+,E_-]=E_0,\qquad [E_0,E_\pm]=\pm E_\pm,\qquad E_a^\dagger=E_{-a}.
 \label{su2rel}
\ee
In fact we can realize \  $L ,x_+, x_-$ \ by setting \cite{FioPisJGP18}
(we simplify the notation dropping $\pi_{{{\Lambda}}}$)
\be
\ba{c}
L=E_0, \qquad  x_\pm=f_{\pm}(E_0)E_\pm,\\[10pt]
\displaystyle  f_{+}(s)=\sqrt{\frac{1\!+\!s(s\!-\!1)/k}{{{\Lambda}}({{\Lambda}}+1)\!-\!s(s\!-\!1)}}= f_{-}(s-1),
\ea\label{transfD2}\
\ee
i.e. in a sense the $x_\pm$ are $E_\pm$ (which play the role of $x_\pm$ in  Madore FS) squeezed in the $E_0$ direction; one can easily check (\ref{commrelD=2'}-\ref{commrelD=2})  using (\ref{azioneL}), 
with $L_a,l,m$ resp.  replaced   by $E_a,\Lambda,n$.
Hence \ $\pi_{\Lambda}(E_+),\pi_{\Lambda}(E_-)$ \  are generators of 
$\A_\Lambda$ alternative to \ $x_+,x_-$.

\item
The group $Y_\Lambda\simeq SU(2\Lambda\!+\!1)$ of $*$-automorphisms  of $\A_\Lambda$ is inner and includes a subgroup $SO(3)$ {\it independent of $\Lambda$}
(acting irreducibly via $\pi_\Lambda$) and a subgroup $O(2)\subset SO(3)$  corresponding to orthogonal transformations (in particular, rotations) of the coordinates $x_i$, which plays the role of isometry group of $S^1_\Lambda$.

\end{itemize}
As in the commutative case we define \ $\left\langle\bx\,\right\rangle^2\!:=\!\langle x_1\rangle^2\!+\!\langle x_2\rangle^2$ and find 
$\left\langle\bx\,\right\rangle^2\!=\!\langle x_+\rangle\langle x_-\rangle\!=\!|\langle x_+\rangle|^2$.

\subsection{Diagonalization of the coordinates $x_i$ on $S^1_\Lambda$}
\label{circlespec}

\noindent
As said, by $O(2)$-covariance $\Sigma_{x_i}\left(\Lambda\right)=\Sigma_{x_1}\left(\Lambda\right)$ for all $i$, so we can study just the spectrum $\Sigma_{x_1}\left(\Lambda\right)$. 
$L$ is invariant under $2$-dimensional rotations, whereas $L\to -L$ under $x_1$- or $x_2$-inversion. 
On the basis $\B_{\Lambda}$ the operator
$x_1$ is represented by the $(2\Lambda\!+\!1)\times(2\Lambda\!+\!1)$ symmetric tridiagonal matrix  
$$
X(\Lambda)=\frac 12\left(
\begin{array}{cccccccc}
0&b_{\Lambda} &0&0&\dots&0&0&0\\
b_{\Lambda}  &0&b_{\Lambda-1} &0&\dots &0&0&0\\
0&b_{\Lambda-1} & 0&b_{\Lambda-2} &\dots&0&0&0\\
\vdots&\vdots&\vdots&\vdots&\ddots&\vdots&\vdots&\vdots\\
\vdots&\vdots&\vdots&\vdots&\ddots&\vdots&\vdots&\vdots\\
\vdots&\vdots&\vdots&\vdots&\ddots&\vdots&\vdots&\vdots\\
%0&0&0&0&\cdots&a&b&0\\
0&0&0&0&\cdots&b_{2-\Lambda} &0&b_{1-\Lambda} \\
0&0&0&0&\cdots&0&b_{1-\Lambda} &0 
\end{array}
\right)         =X^0(\Lambda)+O\!\left(\!\frac {1}{\Lambda^2}\!\right).
$$
Here  $b_n\equiv\sqrt{1\!+\! n(n \!-\! 1)/k}$, \ and $X^0$ is the $k\to\infty$ limit
of $X$, i.e. is obtained replacing all $b_n$ by $1$. The eigenvectors and eigenvalues
of Toeplitz matrices such as $X^0$ are known (see e.g. \cite{Noschese} p. 2-3)
and are good approximations of those of $x_1$; in \cite{FioPis19JPA} we have studied
the latter estimating the needed corrections. The spectrum of $X^0(\Lambda)$ arranged in descending order is \
$\Sigma_{X_0}:=\left\{\widetilde{\alpha}_h(\Lambda)\right\}_{h=1}^{N}$, \ where
\be\label{valuecos}
\widetilde{\alpha}_h=\cos{\left(\frac{h \pi}{N+1} \right)},
%,\quad h=1,2,\cdots,2\Lambda+1.
\ee
and $N=2\Lambda+1$. We shall arrange also the spectrum $\Sigma_{X(\Lambda)}=\left\{\alpha_h(\Lambda)\right\}_{h=1}^{2\Lambda+1}$ of $x_1\simeq X$  in decreasing order. 
Improving the results of Theorem 3.1  in \cite{FioPis19JPA}, here we prove

\begin{teorema}   For all $\Lambda\in\NN$ 
\begin{enumerate}

\item If $\alpha$ belongs to $\Sigma_{X}$, then also $-\alpha$ does.

\item $\Sigma_{X(\Lambda)},\Sigma_{X(\Lambda+1)}$ interlace, i.e. between any two
 consecutive eigenvalues of $X(\Lambda\!+\!1)$ there is exactly one of $X(\Lambda)$  (see fig. \ref{InterLace}):
\be 
\alpha_1\left(\Lambda\!+\!1\right)\: >\:\alpha_1\left(\Lambda\right)\: >\:\alpha_2\left(\Lambda\!+\!1\right)
\: >\:\alpha_2\left(\Lambda\right)\: >...\: >\:\alpha_{\Lambda}\left(\Lambda\right)\: >\:\alpha_{\Lambda+1}\left(\Lambda\!+\!1\right);
\nonumber
\ee

\item $\Sigma_{X}$ becomes uniformly dense in $[-1,1]$ as $\Lambda\to \infty$, in particular \
$\displaystyle \alpha_1\left(\Lambda\right)\geq 1-\frac{\pi^2}{8(\Lambda+1)^2}$.

\end{enumerate}
\label{DiagD2}
\end{teorema} 

\begin{proof} 
Items 1.,  3. were proved in  Theorem 3.1  in \cite{FioPis19JPA}. 
 Item  1. follows also from Proposition \ref{interlace}, after  the inversion $X_m\mapsto-X_m$.
We prove item  2.with the help of Proposition \ref{interlace} (?) COMPLETARE.
\end{proof}

Note that item 2. implies in particular that all  eigenvalues are simple.

In the $\Lambda\to\infty$ the eigenvectors of $x_1$ become generalized eigenvectors,
as expected.

\subsection{\small $O(2)$-covariant uncertainty relations and $O(2)$-invariant strong SCS systems}
\label{circleCS1}

\noindent
From (\ref{commrelD=2'}) one can derive \cite{FioPis19LMP}
for both $S^1,S^1_\Lambda$ the $O(2)$-covariant  `Heisenberg' uncertainty relations  
\bea
%\mbox{HUR:}\quad 
\Delta L\,\Delta x_1\ge \frac {|\langle x_2\rangle|}2,\qquad 
\Delta L\,\Delta x_2\ge \frac {|\langle x_1\rangle|}2,\qquad 
\Delta L^2(\Delta \bx)^2\ge \frac {\langle \bx\rangle^2}4;       \label{HURS^1}
\eea
they are saturated by the $\psi_n$ ($\Delta L=0$). We have  also 
shown that $\Delta x_1,\Delta x_2$  may vanish separately, but not simultaneously, because
\bea
(\Delta\bx)^2\ge(\Delta\bx)^2_{min}\sim {\color{red}{ \frac1{\Lambda^2} }}.
%, \qquad \Delta x_1\, \Delta x_2 \ge  {\color{red}{\frac 12 \left\langle L'\right\rangle}}
%\quad \Leftarrow \:(\ref{y+y-}).
\eea

\medskip
\noindent
{\bf Theorem} (section 3.1 in \cite{FioPis19LMP}) \ {\it
  The system \
${\cal S}^\beta\equiv\left\{\bomega_\alpha^\beta\equiv\!\!\sum\limits_{n=-\Lambda}^{\Lambda}
\!\!\frac{e^{i(\alpha n+\beta_n)}}{\sqrt{2\Lambda\!+\!1}}\psi_n
%\,|\, \alpha\in \Omega\equiv S^1
\right\}_{\!\!\alpha\in \RR/2\pi\ZZ}$ is a strong SCS,
\bea
%\mbox{is a strong SCS:}\qquad
\frac{2\Lambda\!+\!1}{2\pi}\int_0^{2\pi}\!\!\!\! d\alpha \,P^\beta_\alpha=\id,
\qquad P_\alpha^\beta\equiv\bomega_\alpha^\beta\langle\bomega_\alpha^\beta,\cdot\rangle,
%\qquad \bomega_\alpha^\beta\equiv\sum_{n=-\Lambda}^{\Lambda}
%\!\!\frac{e^{i(\alpha n+\beta_n)}}{\sqrt{2\Lambda\!+\!1}}\psi_n.      
%\nonumber
\qquad\label{IdResolS^1_L} 
\eea
for all $\beta\in(\RR/2\pi\ZZ)^{2\Lambda\!+\!1}$ (the label space is
$\RR/2\pi\ZZ\simeq S^1\equiv \Omega$). It
is fully $O(2)$-covariant if $\beta_{-n}=\beta_n$. \
On all $\bomega_\alpha^\beta$ it is $\langle L\rangle=0$, $\left(\Delta L\right)^2=\frac {\Lambda(\Lambda\!+\!1)}{3}$, whereas 
$(\Delta \bx)^2$ is  minimized by the $\bphi_\alpha\equiv\bomega_\alpha^0$, with
\bea
\left(\Delta\bx\right) ^2 <\frac 1{\Lambda+1}
\left(\frac 12+\frac{1}{3\Lambda}\right).
%\overset{\Lambda\ge 2}\le  \frac 2{3(\Lambda+1)}     
\label{utileb} 
 \eea
% in particular $\langle x_2\rangle_{\bphi}=0$,  
%$\langle x_1\rangle_{\bphi}=\langle x_+\rangle_{\bphi}\in\RR$,
%where $\bphi:=\bphi_0=\bomega_0^0$. 
}
Within the class  of strong SCS, the
$\bphi_\alpha$ are closest to classical states(=points)  of $S^1$, and in one-to-one
correspondence with them: \
$S^1\leftrightarrow {\cal S}^1\equiv\{\bphi_\alpha\}_{\alpha\in\RR/2\pi\ZZ\simeq S^1\equiv \Omega}$.

\subsection{$O(2)$-invariant weak SCS  minimizing $(\Delta \bm{x})^2$}
\label{circleCS2}

Since $(\Delta \bx)^2$ is $O(2)$-invariant, so is the set ${\cal W}^1$ of states 
minimizing it; ${\cal W}^1$ is a weak SCS. We can recover the whole set from
any element  $\underline{\bm{\chi}}$ through rotations, \
${\cal W}^1=\left\{\underline{\bm{\chi}}_\alpha\!\equiv e^{i\alpha L}\underline{\bm{\chi}} \right\}_{\!\alpha\in[0,2\pi[}$. Choosing $\underline{\bm{\chi}}$ so that
 $\la x_2\ra_{\underline{\bm{\chi}}}=0$, by (\ref{splitH})  we find
 $\la \bx\ra_{\underline{\bm{\chi}}_\alpha}= \left|\la \bx\ra_{\underline{\bm{\chi}}}\right|\bu_\alpha$, where $\bu_\alpha=(\cos\alpha,\sin\alpha)$.
We have shown that
\bea
0<\left(\Delta{\bx}\right)^2_{min}=\left(\Delta{\bx}\right)^2_{\underline{\bm{\chi}}_\alpha}< \frac{3.5}{(\Lambda+1)^2}.            \label{Deltax2qminS^1_L}
\eea
The (rays associated to) 
$\underline{\bm{\chi}}_\alpha$ are closest to classical states(=points) of $S^1$, 
and in one-to-one correspondence with them: \
$S^1\leftrightarrow {\cal S}^1\equiv\{\bphi_\alpha\}_{\alpha\in\RR/2\pi\ZZ\simeq S^1}$.

\section{D=3:$O\left(3\right)$-covariant fuzzy sphere}
\label{D3}

We use two related sets of  angular momentum and space coordinate operators: the hermitean ones \
$\left\{L_i\right\}_{i=1}^3$  (with  $L_i\equiv\varepsilon^{ijk}L_{jk}/2$) \ and $\left\{x_i\right\}_{i=1}^3$,  and the partly
hermitean conjugate ones \
$\left\{L_a\right\}$, $\left\{x_a\right\}$ (here $a=0,+,-$), \ which are obtained from the former as follows\footnote{Again, here we use the conventions 
of \cite{FioPis19JPA,FioPis19LMP}, rather than those of   \cite{FioPisJGP18}.}:
$$
L_{\pm}:=L_1\pm iL_2,\quad\quad L_0:=L_3,\quad\quad x_{\pm}:=x_1\pm ix_2,\quad\quad x_0:=x_3.
$$
The square distance from the origin can be expressed as $\bx^2:=x_{i}x_{i}=x_0^2+(x_+x_-+x_-x_+)/2$.
As a preferred orthonormal basis of the carrier Hilbert space $\Hi_\Lambda$ we adopt one 
$\B_\Lambda$ consisting of eigenvectors of \ $L_3$, \  $\bL^2=L_{i}L_{i}=L_0^2+(L_+L_-+L_-L_+)/2$,
\be \label{D=3Basis}
\B_\Lambda:=\left\{\psi_l^m\right\}_{l=0,1,...,\Lambda; \ m=-l,...,l},\qquad
\bL^2\psi_l^m=l(l+1)\psi_l^m,\qquad L_3\psi_l^m=m\psi_l^m.
\ee
On the $\psi_l^m$ the $L_a,x_a$ act as follows:
\be\label{azioneL}
L_0\psi_l^m=m\,\psi_l^m,\quad
L_{\pm}\psi_l^m=\sqrt{(l\!\mp\! m)(l\!\pm\! m\!+\!1)}\psi_l^{m\pm 1}\!,
\ee
\bea\label{azionex}
x_a\psi_{l}^m=\left\{\!\!
\ba{ll}
c_l A_{l}^{a,m}\psi_{l-1}^{m+a}+
c_{l+1} B_{l}^{a,m} \psi_{l+1}^{m+a
}&\mbox{ if }l<\Lambda,\\[8pt]
c_lA_l^{a,m}\psi_{\Lambda-1}^{m+a}&\mbox{ if }
%-l\leq m\leq l-2\mbox{ and }
l=\Lambda,\\[8pt]
0&\mbox{otherwise,}
\ea
\right. 
\eea
where 
\bea 
\label{Clebsch}
\ba{l}
\displaystyle A_l^{0,m}:=\sqrt{\frac{(l+m)(l-m)}{(2l+1)(2l-1)}}\hspace{0.1cm},\hspace{1cm} A_l^{\pm,m}:=\pm\sqrt{\frac{(l\mp m)(l\mp m-1)}{(2l-1)(2l+1)}}\hspace{0.1cm}, \\
\displaystyle   B_l^{a,m}=A_{l+1}^{-a,m+a},\hspace{0.8cm}
c_l:= \sqrt{1+\frac{l^2}{k}}\quad 1\le l\le \Lambda,\hspace{0.6cm} c_0=c_{\Lambda+1}=0\hspace{0.1cm},           \ea
\eea
and $k(\Lambda)$ fulfills
$k(\Lambda)\ge\Lambda^2(\Lambda\!+\!1)^2  $.
The $L_i,x_i$  fulfill the following $O(3)$-covariant relations: 
\bea
 x_i^{\dag}=x_i, \qquad 
L_i ^{\dag}=L_i, \qquad [L_i,x_j]=i\varepsilon^{ijh}x_h, \qquad 
\left[\,L_i,L_j\right]=i\varepsilon^{ijh}L_h,\qquad x_iL_i=0, \label{rf3D4}\\[8pt]
[x_i,x_j]={\color{red}\underbrace{ i\varepsilon^{ijh}L_h\left(\!\!-\!\frac{1}{k}\!+\!K\widetilde{P}_{\Lambda}\!\right)}_{Snyder-like}},  \hspace{1.1cm}
\bx^2= 1\:
{\color{red}{+\frac{\bL^2\!+\! 1}{k}-\left[1+\frac{(\Lambda\!+\!1)^2}{k}\right]\!\frac{\Lambda\!+\!1}{2\Lambda\!+\!1}\widetilde P_{\Lambda}}},       \label{xx}\\
%&& [x_i,x_j]=i\varepsilon^{ijh}\left(\!-\frac{I}{k}+K\widetilde{P}_{\Lambda}\!\right)L_h, \hspace{1cm}
%\bx^2= 1+\frac{\bL^2\!+\!1}{k}-\left[1+\frac{(\Lambda\!+\!1)^2}{k}\right]\frac{\Lambda\!+\!1}{2\Lambda\!+\!1}\widetilde P_{\Lambda},   \qquad                           \label{xx}\\[4pt]
\prod_{l=0}^{\Lambda}\left[\bL^2-l(l+1)I\right] =0,\hspace{1.5cm}
\prod_{m=-l}^{l}{\left(L_3-mI\right)}\widetilde{P}_l=0,\hspace{1.5cm} \left(x_{\pm}\right)^{2\Lambda+1}=0; \label{rf3D3}
\eea
here $K=\frac{1}{k}+\frac{1+\frac{\Lambda^2}{k}}{2\Lambda+1}$, $\widetilde{P}_l$ is the projection on the  $\bL^2=l(l+1)$ eigenspace. Again, terms marked red are absent in the commutative case. In the $\Lambda\to\infty$
limit also the non-vanishing ones  will play no role at any fixed energy $E$, as they are proportional to the projection $\widetilde{P}_{\Lambda}$
onto the states with highest energy $E_{\Lambda}\to\infty$; 
(\ref{rf3D3}a,b) give back  the spectra of $\bL^2,L_3$ on ${\cal L}^2(S^2),{\cal L}^2(\RR^3)$, whereas (\ref{rf3D3}c) looses meaning and must be dropped.
We point out that: 
\begin{itemize} 

\item  $\bx^2\neq 1$; but it is a function of $\bL^2$, hence the $\psi_l^m$ are
its eigenvectors;  its eigenvalues  (except when $l=\Lambda$) are close to 1, slightly grow with $l$ and collapse to 1 as $\Lambda\to \infty$.

\item  The ordered monomials in $x_i,L_i$ [with degrees  bounded by (\ref{rf3D4}-\ref{rf3D3})] make up a 
 basis of the $(\Lambda\!+\!1)^4$-dim  vector space \ $\mathcal{A}_{\Lambda}\!:=\!End(\Hi_\Lambda\!)\!\simeq\! M_{(\Lambda+1)^2}(\CC)$, 
because the $\widetilde{P}_l$ themselves can be expressed as polynomials in $\bL^2$.

\item The $x_i$ {\it generate} the $*$-algebra $\mathcal{A}_{\Lambda}$, because
also the $L_i$ can be expressed as non-ordered polynomials in the $x_i$.

\item As anticipated in (\ref{isomD}), actually   there are 
$O(3)$-covariant $*$-algebra isomorphisms 
 \be
\A_{\Lambda}\simeq M_N(\CC)\simeq\bpi_\Lambda[Uso(4)], \quad N:=(\Lambda\!+\!1)^2.             \label{isomD3}
\ee
where $\bpi_{\Lambda}$ is the  $N$-dimensional unitary vector (and irreducible)  representation 
%$\bpi_{\Lambda}=\pi_{\frac{\Lambda}{2}}\otimes  \pi_{\frac{\Lambda}{2}}$   
of $Uso(4)%\simeq U\!su(2)\otimes U\!su(2)
$  on the Hilbert space 
${\bf V}_{\Lambda}%=V_{\frac{\Lambda}{2}}\otimes V_{\frac{\Lambda}{2}}
$  characterized by the conditions $\bpi_{\Lambda}(C)=\Lambda(\Lambda+2)$, 
$\bpi_{\Lambda}(C')=0$ on the quadratic Casimirs.
In terms of the Cartan-Weyl basis  $\{{\sf L}_{HI }\}$ ($H,I \in\{1,2,3,4\}$)  of $so(4)$,
\be
\left[{\sf L}_{HI  },{\sf L}_{J K }\right]=i\left(\delta_{HJ }{\sf L}_{I  K }-\delta_{HK }{\sf L}_{I  J }-\delta_{I  J }{\sf L}_{HK }+\delta_{I  K }{\sf L}_{HJ }\right), \qquad {\sf L}_{HI  }^\dagger={\sf L}_{HI  }=-{\sf L}_{I  H},
 \label{so4rel}
\ee
$C={\sf L}_{I  J }{\sf L}_{I  J }$,  $C'=\varepsilon^{HI  J K }{\sf L}_{HI }{\sf L}_{J K }$  (sum over repeated indices).
To simplify the notation we drop $\bpi_{\Lambda}$. 
In fact one can realize \  $L_i\, ,x_i$, \ $i\in\{1,2,3\}$, by  \ setting \cite{FioPisJGP18}
\be
\ba{rcl}
 L_i&=& \displaystyle\frac 1{2}\varepsilon^{ijk4}{\sf L}_{jk}, \qquad \qquad  x_i=
g^*(\lambda)\, {\sf L}_{4i}\,g(\lambda),\\[10pt]
 g(l) &=& \displaystyle\sqrt{
\frac{\Gamma\!\left(\frac {\Lambda\!+\!l}2\!+\!1\right)\Gamma\!\left(\frac {\Lambda\!-\!l\!+\!1}2\right)}
{\Gamma\!\left(\frac {\Lambda\!+\!1\!+\!l}2\!+\!1\right)\Gamma\!\left(\frac {\Lambda\!-\!l}2\!+\!1\right)}
\frac{\Gamma\!\left(\frac l2\!+\!1\!+\!\frac{i\sqrt{k}}2\right)\Gamma\!\left(\frac l2\!+\!1\!-\!\frac{i\sqrt{k}}2\right)}
{\sqrt{k}\:\Gamma\!\left(\frac {l\!+\!1}2\!+\!\frac{i\sqrt{k}}2\right)\Gamma\!\left(\frac {l\!+\!1}2\!-\!\frac{i\sqrt{k}}2\right)}}\\[22pt]
&=&  \displaystyle\sqrt{\frac{\prod_{h=0}^{l-1}(\Lambda\!+\!l\!-\!2h)}{\prod_{h=0}^l(\Lambda\!+\!l\!+\!1\!-\!2h)}
\prod_{j=0}^{\left[\frac{l\!-\!1}2\right]}\frac{1+\frac{(l\!-\!2j)^2}k}{1+\frac{(l\!-\!1\!-\!2j)^2}k}}\:\:;
\ea\label{transfD3}
\ee
here we have introduced the operator $\lambda:=[\sqrt{4L_iL_i+1}-1]/2$ (which has eigenvalues $l\in\{0,1,...,\Lambda\}$),
 $\Gamma$ is  Euler gamma function, the last equality holds only if $l\in\NN_0$, 
and $[b]$ stands for the integer part of $b$. 
 Therefore  the $L_{HI}$  in the  $\bpi_{\Lambda}$-representation make up also an alternative set of generators of $\A_{\Lambda}$ (in \cite{FioPisJGP18}  ${\sf L}_{4i}$ is denoted by $X_i$).

\item The group $Y_\Lambda\simeq SU(N)$ of $*$-automorphisms  of $\A_\Lambda$ is inner and includes a subgroup $SO(4)$ {\it independent of $\Lambda$}
(acting irreducibly via $\bpi_\Lambda$) and a subgroup $O(3)\subset SO(4)$ corresponding to orthogonal transformations (in particular, rotations) of the coordinates $x_i$, which play the role of isometries of $S^2_\Lambda$.

\end{itemize}

\subsection{Diagonalization of the coordinates $x_i$ on $S^2_\Lambda$}
\label{spherespec}

Again, by $O(3)$-covariance  all $x_i$ have the same spectrum,  so we study the one $\Sigma_{x_3}$
of $x_3\equiv x_0$. Since $\left[x_0,L_0\right]=0$, and $\Sigma_{L_0}$ is known from (\ref{D=3Basis}), we  look for simultaneous eigenvectors of $L_0,x_0$ 
\be \label{diagxL}
L_0\,\bm{\chi}_{\alpha}^m=m\, \bm{\chi}_{\alpha}^m,\qquad
x_0\,\bm{\chi}_{\alpha}^m=\alpha \,\bm{\chi}_{\alpha}^m, \qquad\quad
m=-\Lambda,1\!-\!\Lambda,...,\Lambda
\ee
in the form \ $\bm{\chi}_{\alpha}^m=\sum_{l=|m|}^{\Lambda}{\chi_{\alpha,l}^m\psi_l^m}$. \ The second equation 
\begin{comment}
becomes
\be
\left\{
\begin{split}
\chi_{\alpha,|m|+1}^{m}c_{|m|+1}A_{|m|+1}^{0,m}&=\alpha \chi_{\alpha,|m|}^{m}\\
\chi_{\alpha,|m|}^{m}c_{|m|+1}B_{|m|}^{0,m}+\chi_{\alpha,|m|+2}^{m}c_{|m|+2}A_{|m|+2}^{0,m}&=\alpha \chi_{\alpha,|m|+1}^{m}\\
\chi_{\alpha,|m|+1}^{m}c_{|m|+2}B_{|m|+1}^{0,m}+\chi_{\alpha,|m|+3}^{m}c_{|m|+3}A_{|m|+3}^{0,m}&=\alpha \chi_{\alpha,|m|+2}^{m}\\
\quad\vdots\quad\quad\vdots\quad\quad\vdots\quad\quad\vdots\quad\quad\vdots\quad\quad\vdots\quad\quad&\vdots\quad\quad\vdots\\
\chi_{\alpha,\Lambda-2}^{m}c_{\Lambda-1}B_{\Lambda-2}^{0,m}+\chi_{\alpha,\Lambda}^{m}c_{\Lambda}A_{\Lambda}^{0,m}&=\alpha \chi_{\alpha,\Lambda-1}^{m}\\
c_{\Lambda}B_{\Lambda-1}^{0,m}\chi_{\alpha,\Lambda-1}^{m}&=\alpha \chi_{\alpha,\Lambda}^{m}
\end{split}
\right.
\label{autovalvec}
\ee
which in turn 
\end{comment}
can be rewritten in the matrix form $X_m(\Lambda)\chi=\alpha \chi$,
where $\chi=\left(\chi_{\alpha,|m|}^{m},\chi_{\alpha,|m|+1}^{m},\ldots,\chi_{\alpha,\Lambda}^{m} \right)^T$ and $X_m(\Lambda)$  is the following $N(\Lambda;m)\times N(\Lambda;m)$ \ [with \ $N(\Lambda;m):=\Lambda\!-\!|m|\!+\!1$ ] \ real, symmetric, tridiagonal matrix
$$
\hspace{-0.25cm}
X_m(\Lambda)=\left(\!\!
\begin{array}{cccccccc}
0&\!c_{|m|+1}A_{|m|+1}^{0,m}\!&0&0&0&0&0&0\\
c_{|m|+1}A_{|m|+1}^{0,m}\!&0&\!c_{|m|+2}A_{|m|+2}^{0,m}\!&0&0&0&0&0\\
0&\! c_{|m|+2}A_{|m|+2}^{0,m}\!&0& \!c_{|m|+3}A_{|m|+3}^{0,m}\!&0&0&0&0\\
\vdots&\vdots&\vdots&\vdots&\ddots&\vdots&\vdots&\vdots\\
\vdots&\vdots&\vdots&\vdots&\ddots&\vdots&\vdots&\vdots\\
0&0&0&0&0& \!c_{\Lambda-1}A_{\Lambda-1}^{0,m}\! &0&\! c_{\Lambda}A_{\Lambda}^{0,m}\\
0&0&0&0&0&0&\! c_{\Lambda}A_{\Lambda}^{0,m}\! &0\\ 
\end{array}
\!\!\right).
$$
Since $X_m\left(\Lambda\right)\equiv X_{-m}\left(\Lambda\right)$, we can stick to  $m\in\{0,1,...,\Lambda\}$. We shall arrange the spectrum of $X_m$
 $\Sigma_{X_m}=\left\{\alpha_h\left(\Lambda;m\right) \right\}_{h=1}^{N(\Lambda;m)}$
 in descending order. Improving the results of Theorem 4.1  in \cite{FioPis19JPA} we prove

\begin{teorema}   For all $\Lambda\in\NN$ 
\begin{enumerate}

\item If $\alpha$ belongs to  $\Sigma_{X_m}$,  then also $-\alpha$ does.

\item $\alpha_1\left(\Lambda;0\right)>\alpha_1\left(\Lambda;1\right)>\cdots>\alpha_1\left(\Lambda;\Lambda\right)$
\label{disalpha}

\item $\Sigma_{X_m(\Lambda)},\Sigma_{X_m(\Lambda+1)}$ interlace, i.e. between any two consecutive
eigenvalues of $X_m(\Lambda\!+\!1)$ there is exactly one of $X_m(\Lambda)$  (see fig. \ref{InterLace}):
\be 
\alpha_1\left(\Lambda\!+\!1;m\right)\: >\:\alpha_1\left(\Lambda;m\right)\: >\:\alpha_2\left(\Lambda\!+\!1;m\right)
\: >\:\alpha_2\left(\Lambda;m\right)\: >...\: >\:\alpha_{\Lambda}\left(\Lambda;m\right)\: >\:\alpha_{\Lambda+1}\left(\Lambda\!+\!1;m\right);
\nonumber
\ee

\item $\Sigma_{X_0}$ becomes uniformly dense in $[-1,1]$ as $\Lambda\to \infty$, with $\alpha_1\left(\Lambda;0\right)\geq \displaystyle 1\!-\!\frac{\pi^2}{2(\Lambda+2)^2}$ if $\Lambda\geq 2$.

\end{enumerate}
\label{DiagD3}
\end{teorema} 

\begin{proof} 
Items 1., 2., 4. were proved in  Theorem 4.1  in \cite{FioPis19JPA}. In particular  4.
is based on the fact that most the highest $\alpha_h\left(\Lambda;m\right) $ are well 
approximated by the eigenvalues %$\widetilde{\alpha}_h\left(\Lambda;m\right)$ 
of the Toeplitz matrix $X_m^0$  that is obtained from $X_m$ 
replacing all nonzero elements by $1/2$, although in this case  \ $c_lA_l^{0,m}\nrightarrow1/2$,
\ and \ $c_l A_l^{0,m}\simeq1/2$ \ holds only if \ $|m|\ll l$; the eigenvalues of $X_m^0$ are given by  (\ref{valuecos})  with \ $N=N(\Lambda;m)$. \ \  
Item  2. is a direct consequence of Proposition \ref{interlace}. Also item  1., after  the inversion
$X_m\mapsto-X_m$.
\end{proof}

\noindent
Item 2. implies in particular that all  eigenvalues are simple.

\noindent
As $\Lambda\to\infty$ the eigenvectors of $x_3$ become generalized eigenvectors,
as expected; in particular, the one with the highest eigenvalue $\alpha_1\!\left(\Lambda;0\right)$
becomes a Dirac delta concentrated in the North pole.

\subsection{$O(3)$-invariant  UR and strong SCS  on $S^2_\Lambda$}
\label{sphereCS1}

{\bf Theorem 4.1 in \cite{FioPis19LMP}.} \ \  {\it The uncertainty relation 
\bea
(\Delta\bL)^2\ge |\langle \bL\rangle| \qquad\Leftrightarrow \qquad
\langle \bL^2\rangle
\ge |\langle \bL\rangle|\left(|\langle \bL\rangle|+1\right)  \label{LUR3''}
\eea
holds on $\Hi_\Lambda=\oplus_{l=0}^\Lambda V_l$ and is saturated by
the spin coherent states $\bm{\phi}_{l,g}:=  \bpi_\Lambda (g)\psi^l_l\in V_l$, $l\in \{0,1,..., \Lambda\}$, $g\in SO(3)$. Moreover on $\Hi_\Lambda$
the following resolution of identity holds:
\be
I=\sum_{l=0}^\Lambda C_l\int_{SO(3)} d\mu(g) P_{l,g},\qquad\qquad C_l=\frac{2l\!+\!1}{8\pi^2},\qquad  P_{l,g}=\bm{\phi}_{l,g}
\langle\bm{\phi}_{l,g},\cdot\rangle. %,\qquad \bm{\phi}_{l,g}:= T(g)Y^l_l. 
\label{ResolIdS^2_L}
\ee
}
We can parametrize $g\in SO(3)$,  the invariant measure and the integral over $SO(3)$ through the Euler angles $\varphi,\theta,\psi$:  
\bea
&& g=e^{\varphi I_3}e^{\theta I_2}e^{\psi I_3}\quad\mbox{where }\:
I_3:=\left(\!\!\!\ba{ccc} 0 & 1 & 0\\ -1 & 0 & 0 \\ 0 & 0 & 0\ea\!\!\right)\!,\quad
I_2:=\left(\!\!\ba{ccc} 0 & 0 & -1\\ 0 & 0 & 0 \\ 1 & 0 & 0\ea\!\!\right) 
\quad\Rightarrow\\
&& \bpi_\Lambda(g)=
e^{i\varphi L_3}e^{i\theta L_2}e^{i\psi L_3}, \quad
\int\limits_{SO(3)}\!\!\!\! d\mu(g)=\!\!\int\limits^{2\pi}_0\!\!d\varphi\!\!\int\limits^{\pi}_0\!\!d\theta\sin\theta\!\!\int\limits^{2\pi}_0\!\!d\psi=8\pi^2.   %  \label{eulerangles}
\eea
In  (\ref{ResolIdS^2_L}) integration
over $\psi$ can be actually eliminated rescaling $d\mu$ by $2\pi$, i.e. one can integrate just over $S^2$, because the $\psi^l_l$  are eigenvectors of $L_3$.
The theorem holds also for $\Lambda=\infty$, i.e. on $ {\cal L}^2(S^2)$,
because on the latter the commutation relations  $[L_i,L_j]=i\varepsilon^{ijk}  L_k $  are the 
same:  the UR  (\ref{LUR3''}) is saturated by
the spin coherent states $\bm{\phi}_{l,g}:=  \bpi_\Lambda (g)Y^l_l\in V_l$,   and 
%the resolution of the identity 
(\ref{ResolIdS^2_L}) holds provided $l$ run  over $\NN_0$ and we replace $\psi^l_l$ by $Y^l_l$, $\bpi_\Lambda$ by the (reducible) representation of $SO(3)$ on  $ {\cal L}^2(S^2)$ \cite{FioPis19LMP}. 

Again, $\Delta x_1,\Delta x_2,\Delta x_3$  may vanish separately,  not simultaneously, because
\bea
(\Delta\bx)^2\ge(\Delta\bx)^2_{min}\sim {\color{red}{ \frac1{\Lambda^2} }}
\eea

 Fixed a generic normalized vector  
\ $\bm{\omega}\equiv\sum\limits_{l=0}^{\Lambda}\sum\limits_{h=-l}^{l}
\!\omega_l^h\psi_l^h$, for $g\in SO(3)$ let 
\be
\bm{\omega}_g:=\bpi_\Lambda(g)\bm{\omega}%=\sum_{m=-\Lambda}^{\Lambda}e^{i\alpha m}\omega_m\psi_m
,\qquad\qquad P_g:=\bm{\omega}_g\langle\bm{\omega}_g,\cdot\rangle.
\ee
As the unitary representation $\bpi_\Lambda$ of $SO(3)$ on $\Hi_\Lambda$ is {\it reducible}, more precisely  the direct sum of the irreducible representations $(V_l,\pi_l)$, $l=0,...,\Lambda$, completeness  and resolution of the identity for the system  ${\cal S}^\omega\equiv\left\{\bm{\omega}_g\right\}_{\!g\in SO(3)}$ are not automatic.  ${\cal S}^\omega$ is complete
if for all $l$ there exists at least one $h$ such that $\omega_l^h\neq 0$ (then it is also overcomplete).  Moreover, we have proved

{\bf Theorem 4.2 in  \cite{FioPis19LMP}.} \ \ {\it 
${\cal S}^\omega\equiv\left\{\bm{\omega}_g\right\}_{\!g\in SO(3)}$ is a strong SCS if \ $\sum\limits_{h=-l}^{l}|\omega_l^h|^2  =\frac{2l\!+\!1}{(\Lambda+1)^2}$ 
\ $\forall l$; it is also  fully $O(3)$-covariant if \ $\omega_l^h=\omega_l^{-h}$. 
The following resolution of the identity on $\Hi_\Lambda$ holds:
\bea
\id=\frac{(\Lambda+1)^2}{8\pi^2}\int_{SO(3)}\!\! d\mu(g) \,P_g,\qquad
P_g:=\bomega_g \langle\bomega_g,\cdot\rangle.
%,\qquad \bm{\omega}_g:=\bpi_\Lambda(g)\bm{\omega}
\label{ResolIdS^2_Lomegagen}
\eea
}

We can make the isotropy subgroup $H\subset SO(3)$ nontrivial 
choosing e.g. $\bm{\omega}$ an eigenvector of $L_3$; correspondingly $H=\{e^{i \psi L_3}\, |\, \psi\in\RR/2\pi\ZZ\}\simeq SO(2)$. In particular  
$\bphi^\beta=\sum_{l=0}^{\Lambda}\psi_l^0 e^{i\beta_l}
\frac{\sqrt{2l\!+\!1}}{(\Lambda\!+\!1)}$ (with $\beta\in(\RR/2\pi\ZZ)^{\Lambda\!+\!1}$)
has zero eigenvalue.  Setting \ $\bphi_g^\beta=\bpi_\Lambda(g)\bphi^\beta$,
we find that different rays are parametrized by 
$g=e^{\varphi I_3}e^{i\theta I_2}\in SO(3)/SO(2)$.
Hence (\ref{ResolIdS^2_Lomegagen}) holds also with the (normalized ) integration over just the coset space \ $ SO(3)/SO(2)\simeq S^2$.
Based on eqs. (58-59) of  \cite{FioPis19LMP} we thus find 
%the following Corollary,

\begin{corollario} 
${\cal S}^\beta=\{\bphi_g^\beta\}_{g\in S^2}$ is a family  of fully $O(3)$-covariant, strong SCSs, and 
\bea
\id=\frac{(\Lambda+1)^2}{4\pi}\!\!\int^{2\pi}_0\!\!\!\!\!\!d\varphi \!\!\int^{\pi}_0\!\!\!\!d\theta\,\sin\theta \: P_g^\beta,\qquad
P_g^\beta=\bphi_g^\beta \langle\bphi_g^\beta,\cdot\rangle\label{ResolIdS^2_Lphi}
\eea
for all $\beta\in(\RR/2\pi\ZZ)^{\Lambda\!+\!1}$. On it 
$(\Delta \bL)^2$ is independent of $\beta$, while
$(\Delta \bx)^2$ is smallest  on the $\bphi_g^0$, with
\bea
(\Delta \bL)^2=\frac{\Lambda(\Lambda\!+\!2)}{2}, \qquad\qquad
 \left.(\Delta \bx)^2\right|_{\bphi_g^0}<\frac{1}{\Lambda+1}.\label{LXURphi}
\eea
\end{corollario}
Within the class  of strong SCS, the
$\bphi_g^0$ are closest to classical states(=points)  of $S^2$, and in one-to-one
correspondence with them: \
$S^2\leftrightarrow {\cal S}^2\equiv\{\bphi_g^0\}_{g\in SO(3)/SO(2)\simeq S^2}$.

\subsection{$O(3)$-invariant weak SCS on $S^2_\Lambda$ minimizing $(\Delta \bm{x})^2$}
\label{CS3}

Since $(\Delta \bm{x})^2$ is $O(3)$-invariant, so is the set ${\cal W}^2$ of states   minimizing it; ${\cal W}^2$ is a weak SCS. We can recover the whole set from
any element  $\underline{\bm{\chi}}$ through rotations, \
${\cal W}^2=\left\{\underline{\bm{\chi}}_g\!\equiv 
\bpi_\Lambda(g)\underline{\bm{\chi}} \right\}_{g\in SO(3)}$.
Choosing $\underline{\bm{\chi}}$ so that  $\langle \bx\rangle=|\langle \bx\rangle|\Be_3$ [whence  $\langle x_3\rangle=|\langle \bx\rangle|$,
% $\langle x_\pm\rangle=0$, 
$(\Delta \bm{x})^2=\langle \bx^2\rangle-\langle x_3\rangle^2$], 
we find
 $\la \bx\ra_{\underline{\bm{\chi}}_g}= \left|\la \bx\ra_{\underline{\bm{\chi}}}\right|\bu_g$,
 where $\bu_g=g\bu$. 
We have shown that \ $L_3\underline{\bm{\chi}}=0$.
 This implies that the  isotropy subgroup is $H=\{e^{i \psi L_3}\, |\, \psi\in\RR/2\pi\}\simeq SO(2)$   whence 
${\cal W}^2=\{\underline{\bm{\chi}}_g\equiv\bpi_\Lambda(g)\underline{\bm{\chi}}\}_{g=e^{\varphi I_3}e^{i\theta I_2}\in SO(3)/SO(2)\simeq S^2}$, \ 
$\bu_g=(\sin\theta\cos\varphi,\sin\theta\sin\varphi,\cos\theta)$.
The (rays associated to) 
$\underline{\bm{\chi}}_g$ are closest to classical states(=points) of $S^2$, and 
are in one-to-one correspodence with them: \
$S^2\leftrightarrow {\cal S}^2\equiv\{\bphi_g\}_{g\in S^2}$. 
At order $O(1/\Lambda^2)$ $\underline{\bm{\chi}}$ coincides with the
eigenvector $\widehat{\bm{\chi}}$ of $x_3$ with highest eigenvalue
($L_3\widehat{\bm{\chi}}=0$).
We have shown that
\bea
0<\left(\Delta{\bx}\right)^2_{min}=\left(\Delta{\bx}\right)^2_{\underline{\bm{\chi}}_g}< \frac{11}{(\Lambda+1)^2}.            \label{Deltax2qminS^1_L}
\eea

\section{Outlook, comparison with the literature and final remarks}
\label{Conclu}

Imposing an energy cutoff  $\overline{E}$  may: i) yield a simpler low-energy
approximation  $\overline{\T}$  of a well-defined quantum theory $\T$; 
ii) make sense of $\T$ if
$\overline{\T}$ is well-defined while $\T$ is not (as in the case of UV-divergent QFT);
iii) help
in  figuring out from $\overline{\T}$  a new theory valid also at energies $E>\overline{E}$,
if $\overline{E}$ represents a threshold for new physics not accounted for by $\T$.

Denoting by $\Hi$  the Hilbert space of $\T$,
the cutoff is imposed projecting $\T$ on the Hilbert subspace $\overline{\Hi}$ characterized by energies $E$ below $\overline{E}$. The projected  observables fulfill modified algebraic relations; in particular, space coordinates in general become noncommutative. Thus low energy effective theories with space(time) noncommutativity and
lower bounds for space(time) localization (as expected by any candidate theory of quantum gravity) may all naturally arise from the imposition of 
an energy cutoff. Mathematically, $\overline{E}$ can play the role of deformation 
parameter. If $\overline{\Hi}$ remains finite-dimensional for all (finite) $\overline{E}$, the latter may be replaced  by a discrete parameter like $n=\dim(\overline{\Hi})$, and 
$\T_n\equiv\overline{\T}(n)$ make up a fuzzy approximation of $\T$. 
%The latter may be applied also (or only) to the Hilbert space of internal (gauge) degrees of freedom.
If $\T$ lives on a manifold $M$, and in the Hamiltonian we include a suitable confining potential   $U_n$  with a minimum on a submanifold $N$ of $M$ that becomes sharper and sharper as $n\to \infty$,  we  effectively induce a dimensional reduction to a noncommutative quantum theory on $N$.

In the present paper, after elaborating  the arguments sketched in the previous two paragraphs,
we  have reviewed our application of the latter mechanism for the
construction of  a {\bf $d$-dimensional, $O(D)$-covariant fuzzy sphere} ($d=1,2$),  
i.e. a sequence $\{S^d_\Lambda\}_{\Lambda\in\NN}\equiv \{(\Hi_\Lambda,\A_\Lambda)\}_{\Lambda\in\NN}$
of finite-dimensional, $O(D)$-covariant  ($D=d\!+\!1$)  approximations of quantum mechanics (QM) of a spinless
particle on the sphere  $S^d$; \ \   $\bx^2\gtrsim 1$, and $\bx^2$ essentially collapses to 1 as $\Lambda\to\infty$ \ (see the Introduction).
This result has been achieved imposing  an energy-cutoff 
$\overline{E}= \Lambda(\Lambda\!+\!d\!-\!1)$ on QM of  a spinless particle in $\RR^{D}$  subject to a   confining potential $V(r;\Lambda)$ that has a minimum on the  sphere $r=1$ and becomes sharper and sharper as $\Lambda\to\infty$. 
$\A_\Lambda$ is a fuzzy approximation of the {\it whole algebra of observables} of  the  particle on $S^d$
(phase space algebra), and converges to the latter in the limit
$\Lambda\to \infty$. At least for $D=2,3$, there is an $O(D)$-covariant $*$-isomorphism 
$\A_\Lambda\simeq \pi_\Lambda[Uso(D\!+\!1)]$, 
where $\pi_\Lambda$ is a suitable irreducible %unitary 
representation   of $U\!so(D\!+\!1)$ on $\Hi_{\Lambda}$. The latter is 
a {\it reducible} representation of the subgroup $O(D)$ (and of the $U\!so(D)\subset U\!so(D\!+\!1)$ subalgebra generated by the $L_{ij}$), \ more precisely the direct sum of {\it all} the irreducible representations fulfilling  $L^2\le \Lambda(\Lambda\!+\!d\!-\!1)$.
A similar decomposition holds for the subspace ${\cal C}_\Lambda\subset\A_\Lambda$ 
%spanned by ordered monomials 
of completely symmetrized polynomials in the $x_i$ acting as multiplication operators on $\Hi_\Lambda$. For instance, in the case $d=2$ we find
\bea
 \Hi_\Lambda\simeq
\bigoplus\limits_{l=0}^{\Lambda} V_l, \qquad
%\stackrel{\Lambda\to\infty}{\longrightarrow} ,\\
{\cal C}_\Lambda\simeq\bigoplus\limits_{l=0}^{2\Lambda} V_l.       
  \label{directsum}
\eea
where $(V_l,\pi_l)$ are the irreducible representations of $O(3)$ characterized
by   $\bL^2= l(\!+\!1)$. As $\Lambda\to\infty$ 
these respectively become the decompositions  of  ${\cal L}^2(S^2)$ and of 
 $C(S^2)$ that acts on ${\cal L}^2(S^2)$.

Localization in configuration and angular momentum space can be measured through
the $O(D)$-invariant square uncertainties $(\Delta \bx)^2$ (see section \ref{Localization}) 
and $(\Delta \bL)^2$;  for $d=1,2$ we have determined lower bounds and UR characterizing them. In view of future applications of the models, it is crucial to determine systems of coherent states (SCS) on these $S^d_\Lambda$. Section \ref{SCS} is a coincise introduction to SCS. In sections \ref{circlespec}, \ref{spherespec} 
we have studied the eigenvalue equation of a coordinate  $x_i$ (slightly improving
the results of \cite{FioPis19JPA}) and its relation with the minimization of $(\Delta \bx)^2$ for $d=1,2$;  the states minimizing $(\Delta \bx)^2$
make up a $O(D)$-invariant weak SCS ${\cal W}^d$  (sections \ref{circleCS2},  \ref{CS3}). 
 In sections \ref{circleCS1},  \ref{sphereCS1} we have presented the class of
$O(D)$-invariant, strong SCS, in particular the one ${\cal S}^d$  minimizing  $(\Delta \bx)^2$
within the class. 

Let us  compare  $S^2_\Lambda$ with the seminal fuzzy sphere $S^2_n$ of Madore-Hoppe \cite{Mad92,HopdeWNic}.
The $*$-algebra ${\cal A}_n\simeq M_n(\CC)$ of observables on  $S^2_n$ is generated by hermitean coordinates $x_i$ ($i=1,2,3$) fulfilling
\be
[x_i,x_j]=\frac {i}{\sqrt{l(l+1)}}\varepsilon^{ijk}x_k, \qquad
\bx_2:=x_ix_i=1,\qquad l\in\NN/2, \quad n=2l\!+\!1.                      \label{FS}
\ee
In fact $L_i=x_i\sqrt{l(l+1)}$ make up the standard basis of $so(3)$ in the 
irreducible representation $(\pi_l,V_l)$. Hence the spectrum of all  $x_i$ is  
$\Sigma_{x_i}=\left\{  m/\sqrt{l(l\!+\!1)}\: |\:  m=-l,1\!-\!l,...,l\right\}$.
We note that: 

\begin{enumerate}[label=\roman*)]

\item \ Contrary to (\ref{xx}), eq. (\ref{FS}) are not covariant under the whole $O(3)$, 
in particular under parity $x_i\mapsto -x_i$, but only under $SO(3)$. 

\item  \ Contrary to the $\Lambda\to\infty$ limit of (\ref{directsum}),
in the $l\to\infty$ limit  $\Hi=V_l$ remains irreducible and does  not invade ${\cal L}^2(S^2)$. 

\item \ By Theorems \ref{DiagD2}, \ref{DiagD3}, the spectrum of any coordinate $x_i$ on either $S^2_\Lambda$ or $S^2_n$ fulfills the two properties listed in section \ref{diagx}. 
The former  fulfills also  one not shared by the latter:  the eigenstate 
of $x_3$ with maximal eigenvalue, which is very localized around the North pole of $S^2$, is a   $L_3=0$ eigenstate of $L_3$, see fig. \ref{Vett_tg} right.
As $\Lambda\to\infty$ the latter becomes the generalized eigenstate (distribution) $2\delta(\theta)/\sin\theta\simeq\delta(x_1)\delta(x_2)$ on $S^2$ concentrated on the North pole (here $\theta$ is the colatitude); the classical counterpart of this property is that the classical particle on $S^2$ in the position $\bx=(0,0,1)$ has zero $L_3$ ($z$-component of the angular momentum), because
$$
\vspace{-0.1cm}L_3=\left(\underline{\bm{ L}}\right)_3=\left(\underline{\bx} \times \underline{\bm{p}}\right)_3 =0.
$$
On the contrary, on $S^2_n$ this property   is lost;  as the $x_i$ are obtained by rescaling the $L_i$  there is no longer room for  independent observables playing the role of angular momentum operators. 

\item \ On our fuzzy sphere $S^2_\Lambda$ the states with minimal space uncertainty $(\Delta \bx)^2$ make up a weak SCS ${\cal W}^2$, and
$\displaystyle (\Delta \bx)^2_{{\cal W}^2}< \frac{11}{(\Lambda+1)^2}$; the strong SCS
${\cal S}^2$ with minimal  $(\Delta \bx)^2$ has
$\displaystyle (\Delta \bx)^2_{{\cal S}^2}< \frac{1}{\Lambda+1}$.
Both are smaller than the $\displaystyle (\Delta \bx)^2_{min}=\frac{1}{l+1}$  on Madore FS (adopting the same cutoff $l=\Lambda$).

\end{enumerate}

\noindent
Properties i)-iii) in particular show why in our opinion 
$\{{\cal C}_\Lambda\}_{\Lambda\in\NN}$ can be interpreted  as the space of functions on fuzzy  {\it configuration space} $S^2_\Lambda$,  while $\{\A_n\}_{n\in\NN}$ of Madore-Hoppe should be interpreted only as  the space (actually, the algebra)
of functions on a fuzzy {\it spin phase space} $S^2_n$. As for iv), it would be also interesting to compare distances between two maximally localized states on our $S^2_\Lambda$ (either in ${\cal W}^2$  or in ${\cal S}^2$) and  on  the Madore-Hoppe FS  \cite{DanLizMar14}.

\smallskip
Ref. \cite{Pis20} begins to apply in detail our approach to spheres $S^d$  with $d\ge 3$; this allows a first comparison with the rest of the literature. 
The  4-dimensional fuzzy spheres introduced in \cite{GroKliPre96}, as well as the ones
of dimension
$d\ge 3$  considered in \cite{Ramgoolam,DolOCon03,DolOConPre03}, are  based on $End(V)$, 
where $V$ carries a particular {\it irreducible representation} of both $Spin(D)$ and $Spin(D+1)$ (and therefore of both $Uso(D)$ and  $Uso(D+1)$); as $\bx^2$ is central, it can be set $\bx^2=1$ identically. The commutation relations are also $O(D)$-covariant and Snyder-like. 
The fuzzy spherical harmonics are elements, but do do not close a subalgebra, of $End(V)$,
i.e. the product $Y\cdot Y'$ of two spherical harmonics is not a combination of spherical harmonics.
This is exactly as in our models, i.e.  ${\cal C}_\Lambda$ is a subspace, but not a subalgebra, of
$\A_\Lambda$. (One can introduce a product in ${\cal C}_\Lambda $ by projecting 
the result of $Y\cdot Y'$ 
to the vector space ${\cal C}_\Lambda$, but this will be non-associative; associativity is recovered in the $\Lambda\to \infty$ limit).

In \cite{Ste16,Ste17} the authors consider also  the construction of a fuzzy 4-sphere $S^4_N$
through a {\it reducible} representation of $Uso(5)$ on a Hilbert space $V$ obtained decomposing
 an irreducible representation $\pi$ of $Uso(6)$ characterized by a triple of highest weights
$(N,0,n')$; so $End(V)\simeq  \pi[Uso(6)]$, in analogy with our results. 
The elements $X_i$ of a basis of the vector space \ $so(6)\setminus so(5)$ \
play the role of noncommuting cartesian coordinates.  
Hence, the $O(5)$-scalar $\bx^2=X_iX_i$ is no longer central, but its spectrum is still very close to 1  provided $N\gg n'$, because then $V$ decomposes only in few irreducible $SO(5)$-components, all with eigenvalues of $\bx^2$ very close to 1; 
if $n'=0$ then $\bx^2\equiv 1$ ($V$ carries an irreducible representation of $O(5)$), and one recovers the  fuzzy 4-sphere
of   \cite{GroKliPre96}. On the contrary, in our approach 
 $\bx^2\equiv x_ix_i\simeq 1$  is guaranteed  by adopting as noncommutative Cartesian coordinates the  $x_i=f_1(L^2)X_if_2(L^2)$, with suitable functions $f_1,f_2$, rather than the $X_i$. 

\smallskip
Many other aspects and  applications of the general approach 
described in this paper and of these new fuzzy spheres deserve  investigations.
We hope that progresses  can be reported soon.

%\subsubsection*{Acknowledgments}
%We are grateful to F. D'Andrea and T. Weber for useful discussions.

\section{Appendix}

\begin{figure}[t]
\includegraphics[width=15cm]{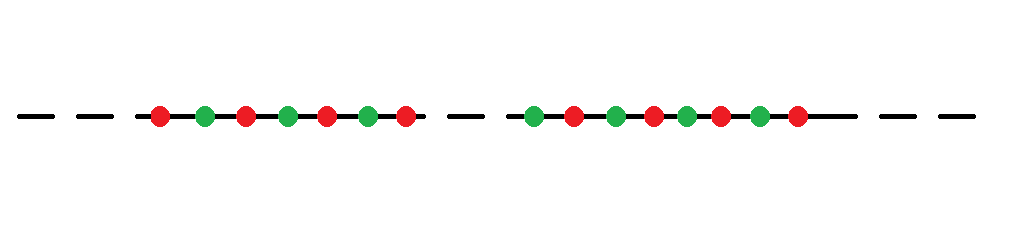}
\caption{The spectra of $A_{n+1}$ (red eigenvalues) and $A_n$ (green eigenvalues) 
interlace.}
\label{InterLace}
\end{figure}

Consider a sequence $\{A_n\}_{n\in\NN}$ of hermitean tridiagonal matrices with zero
diagonal elements %of the form
\be
A_n=\left(\!
\begin{array}{ccccccccc}
0&a_1 &0&0&0&\hdots&0&0&0\\
\overline{a_1} &0& a_2 &0&0&\hdots&0&0&0\\
0&  \overline{a_2} &0&  a_3 &0&\hdots&0&0&0\\
\vdots&\vdots&\vdots&\vdots&\vdots&\ddots&\vdots&\vdots&\vdots\\
0&0&0&0&0&\hdots  &0&  a_{n-2}&0\\
0&0&0&0&0&\hdots&  \overline{a_{n-2}}  &0&  a_{n-1}\\
0&0&0&0&0&\hdots&0&  \overline{a_{n-1}}  &0
\end{array} \!\right).
\ee
For all $n$ the matrix $A_n$ is nested into $A_{n+1}$, more precisely is the upper diagonal
block of the latter.
\begin{comment}
\bea
A_{n+1}=\left(\!
\ba{cc}
A_n                             & \ba{c}0_{n-1}^T\\ a_n \ea \\
0_{n-1} \:\:\:\: \overline{a_n} & 0
\ea \!\right);
\eea
here $0_k$ is the row with $k$ zeroes, $0_k^T$ its transpose column.
\end{comment}
We arrange the (necessarily real) eigenvalues $\alpha^n_h$ of $A_n$ in decreasing order, 
\ $\alpha^n_1\ge\alpha^n_2\ge...\ge\alpha^n_n$. 

\begin{propo} The spectrum $\Sigma_{A_n} =\left\{\alpha^n_h\right\}_{h=1}^{n}$ depends only on the $|a_h|$, $h=1,...,n-1$. For all $n\in\NN$,
$a_h\neq 0$ for all $h=1,...,n$ implies that all   $\Sigma_{A_h},\Sigma_{A_{h+1}}$ interlace, i.e. between any two consecutive
eigenvalues of $A_{h+1}$ there is exactly one of $A_h$ (see fig. \ref{InterLace}).
\label{interlace}
\end{propo}
%\begin{propo} The spectra $\Sigma_{A_n},\Sigma_{A_{n+1}}$ of $A_n,A_{n+1}$ interlace, i.e. between any two consecutive
%elements of $\Sigma_{A_{n+1}}$ there is exactly one element of $\Sigma_{A_n}$.
%\end{propo}
As a particular consequence, all eigenvalues are simple, and the inequalities $\ge$ are strict.
\begin{proof} 
The eigenvalue equation for $A_n$ reads $p_n(\alpha)=0$, where the lhs is the
polynomial of degree $n$ defined by $p_n(\alpha)=\det(\alpha I_n-A_n)$ (here $I_n$ is the unit matrix);  \ $p_n(\alpha)=(\alpha-\alpha^n_1)...(\alpha-\alpha^n_n)$ \ implies that
 $p_n(\alpha)>0$ for all $\alpha>\alpha^n_1$. We easily find $\alpha^1_1=0$ and $\alpha^2_1=|a_1|$, $\alpha^2_2=-|a_1|$, so the claim is true for $n=1$. 
Applying Laplace rule with the last two rows   of the determinant of 
\bea
\alpha I_{n+1}-A_{n+1}=\left(\!\!
\ba{cc}
\alpha I_{n-1}-A_{n-1}\:\:\:  &  \ba{cc}
                                               \ba{c}0_{n-2}^T \\ -a_{n-1} \ea & 0_{n-1}^T\ea   \\
\ba{cc}0_{n-2} & \:\:-\overline{a_{n-1}} \\
0_{n-2} & 0 \ea                                     & \ba{cc}\alpha \:\: & \:\: -a_n\\ 
                                                                -\overline{a_n}\:\:\:\: & \:\:\:\:\alpha \ea
\ea \!\!\right) \nonumber
\eea
(here $0_k$ is the row with $k$ zeroes, $0_k^T$ its transpose column)  we find the recurrence relation \ $p_{n+1}(\alpha)=\alpha p_n(\alpha)-|a_n|^2p_{n-1}(\alpha)$. \ 
Now assume that the claim is true for all $m\le n$, with a generic $n>1$; by the previous relation also $p_{n+1}(\alpha)$, and its roots, depend only on the $|a_h|$, and 
$p_n(\alpha^n_h)=0$ implies
\be    \label{p_nRecurrence}
p_{n+1}(\alpha^n_h)=\alpha^n_h\, p_n(\alpha^n_h)-|a_n|^2
p_{n-1}(\alpha^n_h)=-|a_n|^2
(\alpha^n_h-\alpha^{n-1}_1)(\alpha^n_h-\alpha^{n-1}_2)...(\alpha^n_h-\alpha^{n-1}_{n-1}).
\ee
By the induction hypothesis, 
\be \label{IndHyp}
\alpha^n_1\: >\:\alpha^{n-1}_1\: >\:\alpha^n_2\: >\:\alpha^{n-1}_2\: >\: ...\: >\:\alpha^n_{n-1}\: >\:\alpha^{n-1}_{n-1}\: >\:\alpha^n_n;
\ee
choosing $h=1$ all the brackets at the rhs(\ref{p_nRecurrence}) are positive, and the rhs is negative, hence by continuity there is a $\alpha^{n+1}_1>\alpha^n_1$ such that
$p_{n+1}(\alpha^{n+1}_1)=0$; choosing $h=2$ all the brackets at the rhs are positive but the first one, and the rhs is positive, hence by continuity there is a $\alpha^{n+1}_2\in]\alpha^n_2,\alpha^n_1[$ such that $p_{n+1}(\alpha^{n+1}_2)=0$; ....; finally,
choosing $h=n$ one finds that the sign of  the rhs is $(-1)^{n+1}$, hence by continuity there is a $\alpha^{n+1}_{n+1}<\alpha^n_n$ such that $p_{n+1}(\alpha^{n+1}_{n+1})=0$. 
\end{proof}

\noindent
{\bf Remarks.} \ Inequalities (\ref{IndHyp}) and the sign of $\alpha^{n-1}_h$ determine the sign of 
$$
p_{n+1}(\alpha^{n-1}_h)=\alpha^{n-1}_h\, p_n(\alpha^{n-1}_h)
=\alpha^{n-1}_h(\alpha^{n-1}_h-\alpha^n_1)...(\alpha^{n-1}_h-\alpha^n_n)
$$ 
and whether \ $\alpha^{n+1}_h$ \ belongs to \ $]\alpha^n_h,\alpha^{n-1}_{h-1}]$ \ or \
$]\alpha^{n-1}_{h-1},\alpha^n_{h-1}[$. \ \ On the other hand, if $a_n=0$ then  (\ref{p_nRecurrence})  immediately implies that the $\alpha^n_h$ and $\alpha=0$ are the roots also of $p_{n+1}$, as expected. 

%%%%%%%%%

\end{document}